\documentclass{article}

\usepackage{arxiv}

\usepackage[utf8]{inputenc} %
\usepackage[T1]{fontenc}    %
\usepackage{hyperref}       %
\usepackage{url}            %
\usepackage{booktabs}       %
\usepackage{amsfonts}       %
\usepackage{nicefrac}       %
\usepackage{microtype}      %
\usepackage{lipsum}
\usepackage{graphicx}

\usepackage{booktabs}
\usepackage{diagbox}
\usepackage{eccvabbrv}
\usepackage{arydshln}
\usepackage{amsmath}
\usepackage{graphicx}
\usepackage{algorithm}
\usepackage{algpseudocode}
\usepackage{bm}
\usepackage[normalem]{ulem}
\usepackage{wrapfig}
\usepackage{xcolor}
\usepackage{amsthm}
\newtheorem{theorem}{Theorem}
\newtheorem{lemma}{Lemma}
\usepackage{amssymb}

\title{Provable and Robust Wavefront Sensing via Self-Reference Interferometry}

\author{
Nebiyou Yismaw \\
  University of California, Riverside\\
  \texttt{nyism001@ucr.edu} \\
   \And
Vishwanath Saragadam\\
  University of California, Riverside\\
  \texttt{vishwans@ucr.edu} \\
  \AND
Aswin C. Sankaranarayanan \\
  Carnegie Mellon University\\
  \texttt{saswin@andrew.cmu.edu} \\
  \And
  M. Salman Asif
 \\
 University of California, Riverside \\
  \texttt{sasif@ucr.edu} \\
}

\newcommand{\deltask}[1]{\mathcal{S}_{\Delta_k}\!\big(#1\big)}

\newcommand{\bms}{\bm{s}}
\newcommand{\bmx}{\bm{x}}
\newcommand{\bmy}{\bm{y}}
\newcommand{\bmp}{\bm{p}}
\newcommand{\bmphi}{\bm{\phi}}
\usepackage[capitalize]{cleveref}

\crefname{section}{Sec.}{Secs.}
\Crefname{section}{Section}{Sections}
\crefname{table}{Tab.}{Tabs.}
\Crefname{table}{Table}{Tables}

\usepackage{hyperref}

\begin{document}

\maketitle
\begin{abstract}

Wavefront sensing involves estimating the phase and intensity of light, enabling a wide range of imaging applications, from adaptive optics and astronomy to biomedical imaging. Since conventional image sensors can only measure the spatial intensity distribution, phase retrieval arises as the central problem in wavefront sensing. Conventional interferometric approaches like phase-shifting interferometry (PSI) can recover phase information, but they rely on a stable reference beam that is difficult to realize in practical settings. To overcome this limitation, we propose a novel self-reference framework that relies on interference between shifted copies of the incoming wave; this results in pairwise phase differences between shifted pixels.
We formulate an analytical solution for the complete phase retrieval based on the propagation of these differences across a connected graph. Furthermore, we provide a theoretical analysis of optimal measurement patterns, proving that co-prime shifts guarantee a connected graph and bound worst-case error accumulation, yielding a provably robust method. Extensive simulations demonstrate that complete phase profiles can be recovered from as few as eight shifted measurements, outperforming several existing approaches. Finally, we validate our framework using a hardware prototype, demonstrating real experiments for optical phase profile recovery, auto-refocusing, and imaging through scattering media.

\end{abstract}

\begin{figure}
    \centering
    \includegraphics[width=\linewidth]{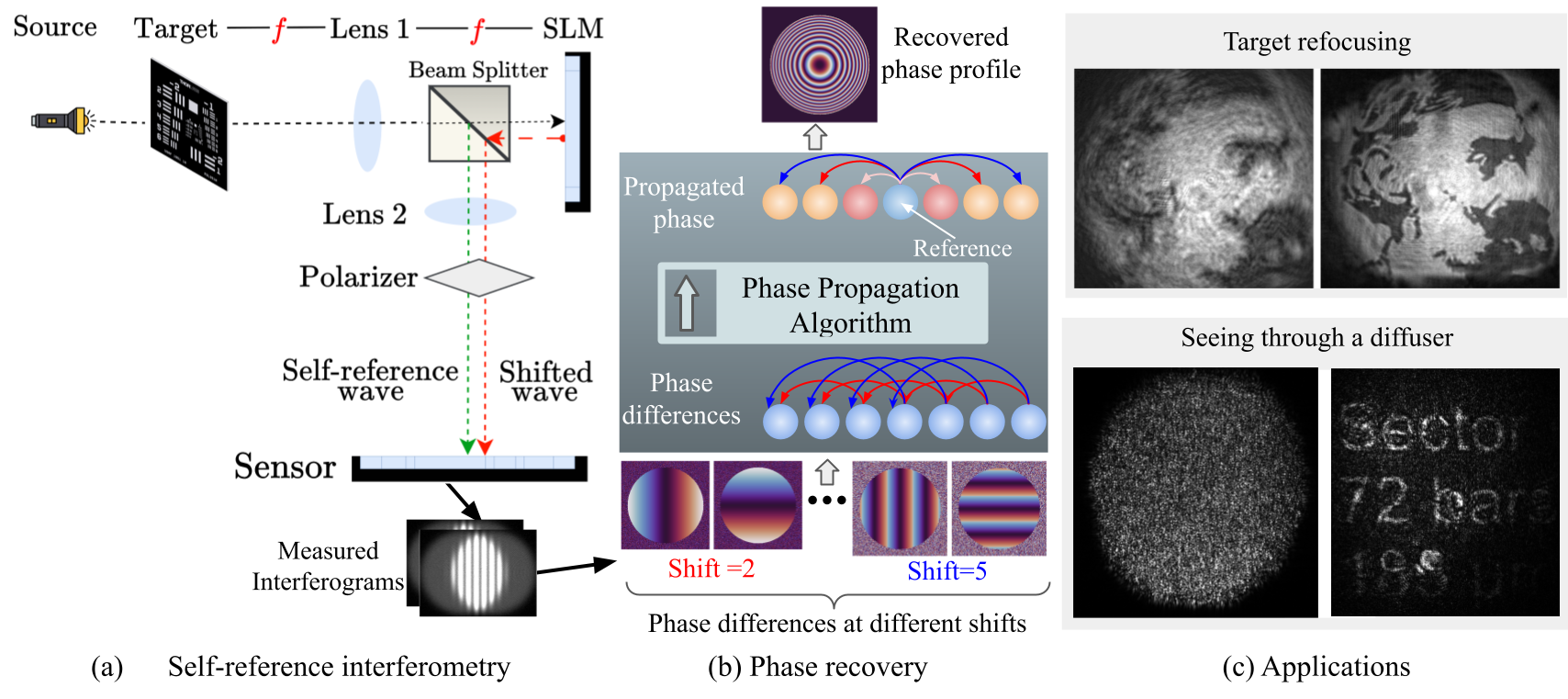}
    \caption{\textbf{Overview of the proposed wavefront sensing via self-reference interferometry.}
(a) Interference of the incoming wavefront with its spatially shifted copy (i.e.\ self reference) provides pairwise phase differences among shifted pixels. (b) Phase differences measured at multiple shifts are propagated using our proposed graph-based algorithm to recover the global phase profile. (c) The recovered phase enables applications such as target refocusing and imaging through a diffuser.}
    \label{fig:intro_figure}
\end{figure}

\section{Introduction}
\label{sec:intro}

Light is an electromagnetic wave. While image sensors measure the spatial distribution of the intensity of this wave, a full description of the wave also involves its phase that is not directly measurable. Phase and intensity provide a complete description of the wave and can enable capabilities ranging from computational refocusing and correction of optical aberrations, as enabled by wavefront sensing approaches, to interferometric systems for imaging microscopic details in scattering medium, enabling holographic displays, and even the detection of gravitational waves \cite{born2013principles, goodman2005introduction, schnars2002digital,  vellekoop2007focusing, park2018quantitative, hampson2021adaptive,booth2014adaptive,huang2024quantitative, abbott2016observation}.

Phase retrieval from intensity-only measurements is a challenging and highly non-convex problem.
The general phase retrieval problem can be defined as recovering a complex-valued signal $\bm{x}$ from its intensity measurements:
\begin{equation}\label{eq:y=|Ax+r|}
    \bm{y} = |\mathcal{A}(\bm{x}) + \bm{r}|^2,
\end{equation}
where $\mathcal{A}(\cdot)$ represents the imaging operator and $\bm{r}$ denotes a known or unknown reference. Structured Fourier transforms are the most common imaging operators in real-world optical systems \cite{goodman2005introduction} but geometric transforms and modulations are also feasible \cite{candes2015phaseCDP,zhang2007phase}. In general, solving the phase retrieval problem requires solving a nonlinear (iterative) optimization problem. Several classical \cite{gerchberg1972practical,fienup1982phase}, model-based \cite{candes2013phaselift,candes2015phase}, and deep learning-based \cite{metzler2018prdeep,icsil2025deep,gandelsman2019double} methods have been proposed for phase retrieval. \textbf{In this work, we seek to develop an analytical, non-iterative, and robust solution for phase retrieval. }

The choice of imaging operator or reference in \eqref{eq:y=|Ax+r|} can simplify the phase retrieval problem. For instance, phase-shifting interferometry (PSI)~\cite{hariharan1987digital, de2011phase, lai1991generalized} records interference of the unknown wave $\bm{x}$ with a known (phase-shifted) reference wave $\bm{r}$. The resulting phase retrieval solution can be written in an analytical form~\cite{chen2021reference, oti2003analysis, andersen2009holographic, medecki1996phase}.
PSI systems require a stable reference wave and protection against environmental perturbations, which makes the PSI setups bulky and difficult to realize in many real-world scenarios~\cite{de1995vibration, wyant2003dynamic, malacara2007optical}. Self-reference PSI~\cite{rhoadarmer2004development} and common-path interferometry~\cite{vakhtin2003common} mitigate these limitations by avoiding external reference waves and reducing environmental sensitivity. Our measurement design closely follows these strategies.

In this work, we present a simple, analytical, and provably robust method for self-reference PSI. \Cref{fig:intro_figure} illustrates our
approach, where spatially shifted copies of the incoming wave interfere with each other, from which we extract pairwise phase differences that are integrated via a graph-based algorithm to recover the complete phase profile.
We summarize our contributions as follows.
\begin{itemize}
    \item \textbf{Self-referenced sensing model.} We introduce a self-referenced interferometric framework that uses spatially shifted copies of the incoming wave to extract pairwise phase differences between shifted pixels. The sensing model is physically realistic and enables a simple phase retrieval solution.
    \item \textbf{Robust and analytical phase retrieval algorithm.} We propose a simple, analytical, and robust solution that propagates pairwise phase differences across the measurement graph to recover the full phase profile. The method can be efficiently implemented using parallel propagation of multiple paths. A final refinement step using least-squares estimation offers improved recovery under noise, recovering complete phase profiles from as few as eight shifted measurements.
    \item \textbf{Optimal measurement design with provable guarantees.} The analytical solution for phase retrieval relies on the propagation of pairwise phase differences across a connected graph. We prove that co-prime shifts in the measurements guarantee graph connectivity and bound worst-case error accumulation, yielding provably optimal shift pairs that minimize error propagation.
    \item \textbf{Experimental validation and hardware prototype.} We validate the proposed framework through extensive simulations and a hardware prototype, demonstrating optical phase profile recovery, auto-refocusing, and imaging through
    scattering media.
\end{itemize}

\section{Related Work}

\subsection{Interferometry-based wavefront sensing}

\textbf{Phase-Shifting interferometry.} Interferometry-based wavefront sensing methods recover the phase of an optical field by superimposing it with a reference beam or with a modified copy of itself and measuring the resulting intensity patterns \cite{creath1988v, bruning1974digital, hariharan2010basics}. It has been used in several adaptive optics~\cite{tyson2022principles} applications, including astronomy \cite{beckers1993adaptive, van2006wm, davies2012adaptive, roddier2004adaptive}, metrology~\cite{yoshizawa2009handbook, zuo2022deep}, and biomedical imaging \cite{booth2007adaptive, ji2017adaptive, park2018quantitative}. A widely used interferometric technique for wavefront sensing is \textit{Phase-Shifting Interferometry (PSI)}~\cite{hariharan1987digital, de2011phase}, where known phase shifts are introduced between a coherent reference beam and the unknown wavefront, and the resulting interference patterns are used to recover the phase. A key limitation of PSI is its reliance on a stable, coherent reference beam, which is not always available or practical~\cite{wyant2003dynamic, malacara2007optical}.

\noindent
\textbf{Common-path interferometry.} To mitigate the sensitivity of PSI-based methods, \textit{common-path interferometry}~\cite{vakhtin2003common, gluckstad2001optimal} ensures that the reference and unknown wavefronts share nearly identical optical paths, causing environmental perturbations to cancel as common-mode noise. The \textit{point diffraction interferometer}~\cite{smartt1975theory} is an example of common-path interferometry in which a small pinhole spatially filters a portion of the unknown beam to generate a reference wave. Extensions of this idea include the \textit{phase-shifting point diffraction interferometer}~\cite{medecki1996phase} and more recent reference-wave design approaches such as ReWave~\cite{chen2021reference}.

\noindent
\textbf{Lateral shearing interferometry.}
Another important class of common-path interferometric techniques is \textit{lateral shearing interferometry} (LSI) \cite{servin2007lateral, velghe2006advanced, robledo2013phase}. In LSI, the wavefront interferes with a laterally shifted copy of itself, which produces interference fringes that encode spatial phase differences between points.  Our measurement model is closely related in that it also relies on controlled lateral shifts and self-interference to encode relative phase differences. 

\subsection{Non-interferometric wavefront sensing}
\textbf{Hardware-based sensors.}
Non-interferometric wavefront sensing methods recover phase from intensity-only measurements via hardware sampling or computational reconstruction. Representative hardware-based approaches include the Shack–Hartmann sensor~\cite{platt2001history} and the pyramid wavefront sensor~\cite{ragazzoni1996pupil,verinaud2005adaptive}. Such methods are limited by spatial resolution and noise sensitivity~\cite{tyson2022principles}.

\noindent
\textbf{Computational phase retrieval.}
Computational approaches such as phase retrieval methods reconstruct the wavefront from intensity measurements at one or more planes. Phase retrieval is challenging because it is nonlinear and nonconvex. A wide range of algorithms have been proposed to solve it~\cite{fienup1982phase, candes2013phaselift, netrapalli2013phase}. Classical iterative methods such as Gerchberg–Saxton~\cite{gerchberg1972practical} alternate between spatial and Fourier domains to recover phase consistent with measured intensities. More recent approaches include coded diffraction pattern (CDP) formulations~\cite{candes2015phaseCDP}, convex relaxations such as PhaseLift~\cite{candes2013phaselift}, and gradient-based nonconvex methods such as Wirtinger Flow~\cite{candes2015phase} and alternating minimization schemes~\cite{netrapalli2013phase}. 

\noindent
\textbf{Programmable optical modulation.}
Programmable optical modulation has been used to improve non-interferometric wavefront sensing~\cite{katkovnik2017computational, zhang2007phase}. Several works use spatial light modulators (SLMs) to encode the incoming field with sequential phase 
patterns and recover the complex wavefronts from intensity-only 
measurements~\cite{kohler2009characterization,falldorf2010phase}. WISH~\cite{wu2019wish} uses random phase modulation and alternating-projection phase retrieval~\cite{fienup1982phase} to reconstruct the field.

\noindent
\textbf{Deep learning-based phase retrieval.}
Deep learning has been applied to phase retrieval in several ways. Supervised networks directly predict phase from intensity measurements~\cite{rivenson2018phase, wang2024use}, while model-based approaches use learned denoisers within iterative schemes, including prDeep~\cite{metzler2018prdeep} with RED~\cite{romano2017little} and plug-and-play formulations~\cite{kamilov2017plug,icsil2025deep}. Untrained priors, such as Deep Image Prior (DIP) and Double-DIP, exploit convolutional architectures without external training~\cite{ulyanov2018deep,gandelsman2019double}. Generative models and diffusion priors further regularize phase retrieval by enforcing consistency with intensity measurements~\cite{shoushtari2023diffusion,kaya2025ddrm,kaya2025i2i,chung2022diffusion}, typically applying the prior to real-valued images.

\section{Self Interference and Phase Propagation}
\subsection{Self interference and phase recovery}

Suppose we represent the unknown complex-valued wave at some given plane as a 2D array $\bmx = |\bmx|e^{j\bmphi} \equiv |\bmx|\odot \bmp$, where $\bmphi$ denotes the phase, $\bmp=e^{j\bmphi}$ denotes the phase component or unit phasor, and $\odot$ represents element-wise multiplication.

\noindent \textbf{Self-reference interferometry} treats parts of the incident wave $\bmx$ as a reference (e.g., point diffraction interferometry \cite{smartt1975theory,medecki1996phase,akondi2014digital}). 
We can use a single bright point in $\bmx$ as a reference, e.g., $\bmx(0)$, and capture measurements with multiple phase shifts, $\varphi_q \in \{0,\pi/2, \pi, 3\pi/2\}$, as 
\begin{equation}\label{eq:PS-PDI}
\bmy_q = | \bmx + \bmx(0) e^{j\varphi_q} |^2.
\end{equation}
Since we can only recover the phase of $\bmx$ up to a global ambiguity, we select the phase of $\bmx(0)$ as a reference and set to zero (\ie $\bmx(0) = |\bmx(0)|$, which is recovered directly from $\bmy_q(0)$). Subsequently, we can recover the complete wave $\bmx$ as 
\begin{equation}\label{eq:intrf-reconst}
    \bmy_q = |\bmx|^2 + |\bmx(0)|^2 + 2\text{Re}(\bmx |\bmx(0)| e^{-j\varphi_q}) \quad \Rightarrow  \quad \bmx = \tfrac{1}{4 |\bmx(0)|}\sum_q \bmy_q e^{j\varphi_q}.
\end{equation} 
A single point-based interference provides a closed-form solution for phase retrieval, but also suffers from low signal-to-noise ratio. In this paper, we seek to find similar simple solutions for different self-interferometric measurements. \medskip

\noindent \textbf{Geometric transforms for self reference.} Suppose we can interfere $\bmx$ with its shifted copy along horizontal, vertical, and any other direction. Furthermore, we can capture multiple measurements with different amounts of spatial shifts. We can represent such shifted interferometric measurements as 
\begin{align}\label{eq:PSI_main}
\bmy_{k,q} &= |\bmx + e^{j\varphi_q}\mathcal{S}_{\Delta_k}(\bmx)|^2 + \bm{\eta}_{k,q} \\
&\equiv |\bmx|^2 + |\mathcal{S}_{\Delta_k}(\bmx)|^2 + 2\text{Re}\left(|\bmx|\odot |\mathcal{S}_{\Delta_k}(\bmx^*)| e^{j(\bmphi-\mathcal{S}_{\Delta_k}(\bmphi)-\varphi_q)}\right) + \bm{\eta}_{k,q}, \notag
\end{align}
where $\Delta_k$ encodes the 2D shift vector for $k=1,\ldots,K$, $\mathcal{S}_{\Delta_k}(\cdot)$ represents an operator that shifts input by $\Delta_k$ pixels, $\varphi_q = \frac{\pi}{2} q$ represents the quadrature phases with $q \in \{0, 1, 2, 3\}$, and 
$\bm{\eta}_{k,q}$ represents the measurement noise. While the system in \eqref{eq:PSI_main} does not provide a closed-form solution like the one in \eqref{eq:intrf-reconst}, we can recover the wave using a simple (graph propagation-based) solution. \medskip 

\noindent \textbf{Phase difference propagation.} Our key observation is that we can compute the pairwise phase differences from the measurements in \eqref{eq:PSI_main}, which can then be combined to recover the complete phase profile. In particular, we can estimate the following phasors from the quadrature measurements in \eqref{eq:PSI_main}:
\begin{equation}\label{eq:phase-diff}
    \sum_q \bmy_{k,q}e^{j\varphi_q} \approx 4 |\bmx|\odot |\mathcal{S}_{\Delta_k}(\bmx^*)| e^{j(\bmphi-\mathcal{S}_{\Delta_k}(\bmphi))} ~ \Rightarrow ~ \bmp_k = \tfrac{\sum_q \bmy_{k,q}e^{j\varphi_q}}{|\sum_q \bmy_{k,q}e^{j\varphi_q}|}\approx e^{j(\bmphi-\mathcal{S}_{\Delta_k}(\bmphi))}. 
\end{equation}
In other words, we can estimate unit phasors $\bmp_k$ that encode phase differences between all pixels shifted by $\Delta_k$.  
In the presence of noise, the estimated $\bmp_k$ in \eqref{eq:phase-diff} will be a noisy unit phasor of pairwise phase differences. 
We can assume one pixel as a reference with zero phase (e.g., $\bmphi(0) = 0, \bmp(0)=1$), and compute the relative phase at any other pixel that is connected to the reference pixel by propagating the phasors in \eqref{eq:phase-diff} to the reference. 
\textbf{For robust and complete phase recovery, we  need measurements that yield a graph with every pixel connected to the reference pixel with minimal path length.}

\subsection{Single shift and single-path propagation}
\label{sec:single_shift_prop}
Let us first consider measurements with a single-pixel shift, as illustrated in \cref{fig:single_vs_two_shifts}(a) for a single row. Suppose we record $\bmy_{1,q}$ with $\Delta_1 = (1,0)$ as a single-pixel horizontal shift. We can calculate the phasors for the phase differences between adjacent pixels using \eqref{eq:phase-diff} as $\bmp_1$. We can then start with the reference pixel phase $(\bmp(0)=1)$ and estimate phasors for all the connected pixels as $\widehat{\bmp}(1) = \bmp_1(0)$, $\widehat{\bmp}(2) = \widehat{\bmp}(1) \bmp_1(1)$, and $\widehat{\bmp}(i+1) = \widehat{\bmp}(i)\bmp_1(i)$ for any $i$. 
In other words, we can propagate the phasors to all the connected pixels as $\widehat{\bmp}(i+\Delta_1) = \widehat{\bmp}(i) \bmp_1(i)$.

A single pixel shift in the horizontal and vertical directions would provide a connected graph for a 2D array. In principle, we can propagate the phasors for phase differences and estimate the complete phase profile. In practice, the phase differences will be estimated from noisy measurements, which will lead to noisy estimates. The single-pixel shift-based measurements under noisy conditions are prone to error accumulation (as shown in~\cref{fig:mae_vs_hop_length}). Let us assume that each phase difference estimation incurs an independent error $\epsilon$ with zero mean and variance $\sigma^2$. In the case of single-pixel shifts, the error accumulation at any pixel will be proportional to its distance from the reference pixel.

\subsection{Two shifts and multi-path propagation}
\begin{figure}[t]
    \centering
    \includegraphics[width=0.95\linewidth]{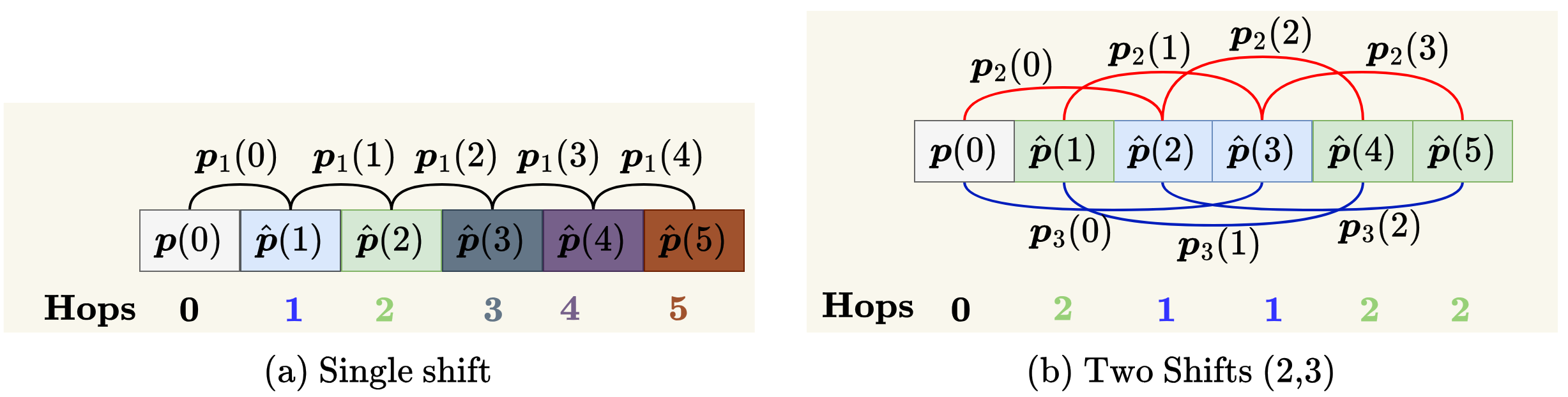}
    \caption{Comparison of phasor propagations by single- and two-shift measurements for $N=9$. Note that $\bmp(i) = e^{j\bmphi(i)}$ denotes the unit phasor of original phase at node $i$, whereas $\bmp_k(i)$ denotes the unit phasor that encodes phase difference between node $i$ and $i+\Delta_k$. 
(a) A single shift ($\Delta_1=1$) forms a path graph with linearly increasing hop distance from the reference. 
(b) Two shifts ($\Delta_2=2$, $\Delta_3=3$) introduce long-range connections that reduce the maximum hop count and shorten propagation paths.}
    \label{fig:single_vs_two_shifts}
\end{figure}

Multiple shifts can provide multiple propagation paths and reduce effective path lengths and error accumulation in the phasor propagation.
Let us consider two distinct shifts $\Delta_s,\Delta_t$ and the  corresponding (estimated) phase-difference phasors $\bmp_s,\bmp_t$. 
For instance, in \cref{fig:single_vs_two_shifts}(b), we choose  $\Delta_s=2,\Delta_t=3$ and compute $\bmp_2,\bmp_3$. 
Assuming $\bmp(0) =1$, we can estimate the phasors at node 2,3 that are directly connected to 0 as $\widehat\bmp(2) = \bmp_2(0), \widehat\bmp(3) = \bmp_3(0)$. We can then propagate these phasors to estimate $\widehat \bmp(4) = \widehat{\bmp}(2) \bmp_2(2),\widehat \bmp(5) = \widehat\bmp(2) \bmp_3(2), \widehat \bmp(1) = \widehat\bmp(3) \bmp^*_2(0).$ Note that to estimate phasor at node 1, we use the estimate at node 3 and then the phasor of difference from 3 to 1. 
The maximum path length (or hop distance) reduces to $2$, which is a significant improvement over the maximum path length of 4 in the single-shift case.

We can reduce the path length of phasor propagation by carefully selecting the shifts $\Delta_s$ and $\Delta_t$ as multiple shifts introduce new connections between the reference pixel and the rest of the graph. These connections allow information to propagate in larger steps and fewer hops compared to the single-shift case. 
The next logical question is how to choose $\Delta_s$ and $\Delta_t$ to ensure complete and robust propagation. We discuss this in~\cref{sec:theory}, where we prove that co-prime shifts guarantee full connectivity and minimum path length.

\begin{algorithm}[t]
\caption{Phase Retrieval for Self-Reference Interferometry}
\label{alg:bfs_phase}
\begin{algorithmic}[1]
\Require Phase-difference phasors $\{\bmp_k\}_{k=1}^K$, shifts $\{\Delta_k\}_{k=1}^K$, reference node $i_0$
\State Initialize phasor estimate $\widehat\bmp$ with reference node $\widehat\bmp(i_0) = 1$
\State Initialize queue $\mathcal{Q} \gets \{i_0\}$
\State Initialize hop map $h(i) \gets \infty$; set $h(i_0) \gets 0$ \hfill {\color{magenta} \it $\triangleright$Track hop length from the reference}
\While{$\mathcal{Q} \neq \emptyset$}
    \State Pop $i \gets \text{dequeue}(\mathcal{Q})$ 
    \For{each neighbor $i'$ reachable via any shift in $\{\Delta_k\}$}
        \State $\bmp' \gets \widehat \bmp(i) \cdot \bmp_k(i)$ \hfill {\color{magenta} \it $\triangleright$ Apply phasor difference}
        \If{$h(i') = \infty$} \hfill {\color{magenta} \it $\triangleright$ Node has not been visited}
            \State Set phase: $\widehat\bmp(i') \gets \bmp'$ 
            \State Set $h(i') \gets h(i)+1$ \hfill{\color{magenta} \it $\triangleright$ Increment hop length}
            \State Enqueue $(i')$ \hfill {\color{magenta} \it $\triangleright$ Push node to queue}
        \ElsIf{$h(i') = h(i) + 1$} \hfill {\color{magenta} \it $\triangleright$ Node was visited with the same hop length}
            \State $\widehat\bmp(i') \gets \text{avg}(\widehat\bmp(i'), \bmp')$ \hfill {\color{magenta} \it $\triangleright$ Average equal hop phasors}
        \EndIf \hfill {\color{magenta} \it $\triangleright$ Discard any phasors with longer hop length}
    \EndFor
\EndWhile
\State \Return phase $\widehat{\bmphi} = \angle \widehat \bmp$ \hfill {\color{magenta} \it $\triangleright$ Compute phases from phasors}
\end{algorithmic}
\end{algorithm}

\subsection{Multiple shifts and multi-path propagation}

The general phase recovery process with $K$ shifts $\{\Delta_k\}_{k=1}^K$ can be formulated as the propagation of phase-difference phasors $\{\bmp_k\}_{k=1}^K$ over a graph whose vertices correspond to pixels and whose edges correspond to the estimated phase-difference phasors, as shown in \cref{fig:single_vs_two_shifts} for 1 or 2 shifts. For complete connectivity, the shifts must be co-prime (see \cref{sec:connectivity}). 

A general algorithm for our method is presented in \cref{alg:bfs_phase}. The method takes $\{\Delta_k\}$, $\{\bmp_k\}$, and a reference pixel as input. The algorithm maintains a queue of nodes $\mathcal{Q}$, which operates in a breadth-first search manner and provides starting nodes for phase propagation at each iteration. The queue is initialized with the reference node. In the first iteration, the phasors of the neighbors reachable from the reference through each shift are computed as $\widehat{\bmp}(\Delta_k) = \bmp(0)\, \bmp_k(\Delta_k)$ and added to the queue. The algorithm then iteratively dequeues elements and propagates phasors from each node to its neighbors. To ensure consistency, a node is only updated if it has not been visited before. Additionally, to improve robustness, propagated phasors from multiple paths are averaged when the paths have the same distance from the reference node.

\noindent
\textbf{Efficient parallel implementation.}
The proposed method shown in~\cref{alg:bfs_phase} processes one pixel at a time, which is inherently sequential. For parallel GPU implementation, we reformulate the propagation as an iterative wavefront expansion, where at every iteration, all the visited pixels simultaneously propagate their phase estimates to all valid neighbors via vectorized tensor operations. We present details of the implementation in the supplementary material.

\noindent
\textbf{Noise reduction via hop-constrained path averaging.} We can reduce the variance in the phase estimate by averaging multiple non-redundant propagation paths that reach the same node/pixel. 
Since the estimation error accumulates with path length (see \cref{sec:single_shift_prop}), we restrict averaging to paths with equal hop lengths and discard paths with larger hop lengths (as they potentially represent larger accumulated noise). This hop-constrained averaging reduces variance without introducing additional bias. This strategy allows us to introduce an additional design axis for propagation. 

\noindent
\textbf{Global least-squares refinement.}  
We can further refine $\widehat{\bmphi}$, estimated using multi-shift propagation~(\cref{alg:bfs_phase}), by jointly fitting it to all the measured phase differences via a global least-squares step. Let $\nabla_k \bmphi = \angle \bmp_k$ denote the wrapped phase differences estimated in \eqref{eq:PSI_main}. We first estimate the integer offsets as 
$m_{k}
= \text{round}\left(
(\deltask{\widehat{\bmphi}}
-
\nabla_k \bmphi
)/2\pi
\right),$
and unwrap the phase difference measurements as
$\nabla_k \bmphi^{\text{unwrap}} = \nabla_k \bmphi+2\pi m_k.$
The refined phase can then be obtained by solving
$
\min_{\bmphi}
\sum_{k=1}^K
\left\|\deltask{\bmphi}
-
\nabla_k \bmphi^{\text{unwrap}}
\right\|_2^2.
$
Since $\deltask{\cdot}$ are shift operators, it can be viewed as a convolution of $\bmphi$ with a filter $\bms_{\Delta_k}$, which can be diagonalized in the Fourier domain. The least-squares solution thus has a simple closed-form expression in the Fourier domain:
\begin{equation}
    \widehat{\bmphi}_{LS} =
    \mathcal{F}^*
    \left(
    \frac{
    \sum_k (\overline{{\mathcal{F}(\bms_{\Delta_k})}}\odot\mathcal{F}(\nabla_k\bmphi^\text{unwrap})
    )}{
    \sum_k |\mathcal{F}(\bms_{\Delta_k})|^2 + \lambda
    }
    \right),
\end{equation}
where $\mathcal{F} $ is the Fourier transform operator, division and multiplications are element-wise, and $\lambda$ is a regularizer.

\section{Provable and Optimal Sensing with Co-Prime Shifts}\label{sec:theory}

A complete and stable phase recovery depends on the graph connectivity and maximum path length with respect to the reference node. 
The choice of shifts $\{\Delta_k\}$ determines the graph connections and these properties. 
We seek to select $\{\Delta_k\}$ that guarantee \textbf{full connectivity} such that every pixel is connected to the reference node through some path and limit the \textbf{maximum path length} that in turn limits the worst-case hop distance and the associated error accumulation. 
In this section, we show that co-prime shifts guarantee full connectivity and minimum path length.

We model the phase recovery problem as propagation over a graph of pixels, as illustrated in~\cref{fig:single_vs_two_shifts}. For simplicity, let us assume the graph consists of $N$ pixels/nodes that are indexed as 
\begin{equation}
V = \left\{-\left\lfloor\tfrac{N}{2}\right\rfloor, \dots, \left\lfloor\tfrac{N}{2}\right\rfloor - 1\right\}.
\label{eq:graph_nodes}
\end{equation} 
An \emph{edge} exists between node $i$ and $j$ if we have direct phase difference measurements between them as $\nabla_k \bmphi$. For a set of shifts $\mathcal{S} = \{s_1, s_2, \dots, s_K\}$, the edge set is
\begin{equation}
    \label{eq:graph_edges}
    E(\mathcal{S}) = \bigl\{ (i, i \pm s_k) : i \pm s_k \in V,\; s_k \in \mathcal{S} \bigr\}.
\end{equation}
Phase propagation can then be viewed as a traversal along this graph. Starting from a reference node, which is assumed to have known or zero phase, we integrate along edges to estimate phases at all other nodes. The \emph{hop distance} between two nodes is defined as the minimum number of edges that must be traversed to connect them. The \emph{diameter} of the graph, which is the maximum hop distance from the reference node to any other node, directly determines the worst-case accumulation of measurement noise.

\subsection{Connectivity guarantees}
\label{sec:connectivity}

We now formalize the conditions under which a graph defined by two shifts of size $s$ and $t$ is guaranteed to be fully connected. In~\cref{theorem:connectivity}, we establish that if the shifts $s$ and $t$ are co-prime (\ie $\gcd(s,t)=1$) and satisfy $s + t \le N$, then every node in $V$ is reachable from the reference node $0$ via a sequence of hops along edges. This result relies on two key observations:~\cref{lemma:co-prime_cover_all_residues} shows that by taking repeated steps of size $s$, we can generate all residues of $\mod (s+t)$, and~\cref{lemma:sliding_window} shows that there is a sliding window of locally connected nodes that connects every node to the reference. The complete proofs of these lemmas and~\cref{theorem:connectivity} are provided in the supplementary material.

\begin{lemma}[Co-prime residue coverage]
\label{lemma:co-prime_cover_all_residues}
Let $k\ge 1$, $s\in\mathbb{Z}$ and $\gcd(s,k)=1$. Then for any $p\in\mathbb{Z}$,
\[
\{ (p + is)\bmod k \;:\; i=0,1,\dots,k-1 \,\}=\{0,1,\dots,k-1\}.
\]
\end{lemma}

\begin{lemma}[Sliding window coverage]
\label{lemma:sliding_window}
Let $s, t \in \mathbb{N}$ with $s + t \le N$ and $V$ as defined in~\cref{eq:graph_nodes}. For any $p \in V$, there exists an interval $I_1$ of size $s+t$ containing $p$, and a sequence of intervals $I_1, I_2, \dots, I_K$ of size $s+t$, each contained in $V$, such that consecutive intervals share at least one common element, and $I_K$ contains $0$.
\end{lemma}

\begin{theorem}[Connectivity via co-prime shifts]
\label{theorem:connectivity}
Let $s,t \in \mathbb{N}$ be co-prime (\ie $\gcd(s,t)=1$), the graph $V$ as defined in~\cref{eq:graph_nodes}, and edges $E\left(\{s,t\}\right)$ as defined in~\cref{eq:graph_edges} using shifts $s$ and $t$. Then every node in $V$ is reachable from the reference node through a sequence of hops along edges, provided $s + t \le N.$
\end{theorem}

\subsection{Optimal co-prime shift design}

With connectivity established, we now determine shift pairs that minimize the maximum hop distance. \cref{lemma:min_hops} establishes a universal lower bound on the maximum hop distance and~\cref{theorem:optimal_coprime_pair} shows that co-prime shifts near $\sqrt{N/2}$ achieve this bound.

\begin{lemma}[Lower bound on maximum hop distance]
    \label{lemma:min_hops}
    Consider the graph $V$ with $N$ nodes defined in~\cref{eq:graph_nodes} and edges $E(\{s,t\})$ as in~\cref{eq:graph_edges}. Let $h^\star$ denote the minimum number of hops required to reach all nodes from a reference node. Then, for any choice of two shifts,
    \begin{equation}
    h^\star \ge \left\lceil \tfrac{-1 + \sqrt{2N - 1}}{2} \right\rceil.
    \end{equation}
\end{lemma}

\begin{theorem}[Optimal shift pair]
    \label{theorem:optimal_coprime_pair}
    For a graph with $N$ nodes as in~\cref{eq:graph_nodes} and edges defined by two shifts $s,t\in\mathbb{N}$, $
    s = \left\lfloor \sqrt{\frac{N}{2}} \right\rfloor, 
    t = \left\lfloor \sqrt{\frac{N}{2}} \right\rfloor + 1,$
    achieve the minimum possible maximum hop distance $h^\star$ given in~\cref{lemma:min_hops}.
\end{theorem}

\subsection{Noise analysis and optimality}

We conduct simulations to validate our analysis that reconstruction error increases with hop length and that optimal shifts~\cref{theorem:optimal_coprime_pair} minimize error accumulation. We generate noisy phase differences~\cref{eq:phase-diff} by adding Gaussian noise to ground-truth phases and recover the phase using \cref{alg:bfs_phase}.

As shown in the left panel of \cref{fig:mae_vs_hop_length} for $N=512$, the cumulative phase error increases monotonically with hop distance, confirming that longer propagation paths accumulate more noise. The optimal shifts $s=\lfloor\sqrt{N/2}\rfloor=16$ and $t=s+1=17$ achieve the minimal $h^\star=16$ and yield the lowest error curve. The right panel of \cref{fig:mae_vs_hop_length} further shows that the hops for our optimal shifts are tightly concentrated, while other choices produce larger maximum hops. We provide additional results in the supplementary material that support our analysis.

\begin{figure}[t]
    \centering
    \includegraphics[width=0.9\linewidth]{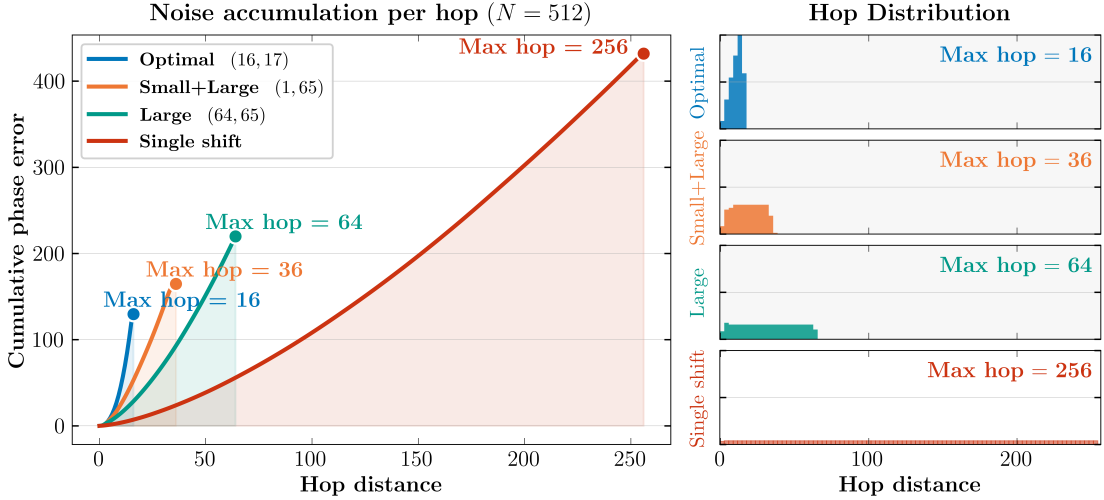}
    \caption{\textbf{Cumulative phase error vs. hop length} for $N=512$ (left) and hop-length histograms (right) for different shift choices. The optimal co-prime pair $(16,17)$ minimizes the maximum hop length and yields the lowest cumulative error. Other shift choices increase the maximum hop length and lead to higher error.}
    \label{fig:mae_vs_hop_length}
    \vspace{-1em}
\end{figure}

\section{Experiments}

\subsection{Simulation experiments}
We evaluate phase recovery from shifted interferometric measurements in a synthetic setting and compare reconstruction accuracy with related baseline methods. To generate simulated data, we use 100 images from the DIV2K validation dataset~\cite{agustsson2017ntire} as ground truth intensity patterns. For each image, we synthesize complex wavefronts by combining the intensity with three types of phase profiles, including random phase, quadratic phase, and smooth phase profiles\footnote{The MATLAB \texttt{peaks} function is used to generate smooth synthetic profiles.}. We simulate interferometric measurements using \cref{eq:PSI_main} with Poisson and Gaussian noise at multiple SNR levels and multiple shifts. We used the optimal shifts $\Delta_{1} = 16$ and $\Delta_{2} = 17$ for the 16-measurement setup (\ie 2 shifts × 2 directions (H, W) × 4 quadratures), and included $\Delta_3=23$ and $\Delta_4 =31$ for the 32-measurement experiments.

We compare against gradient descent–based phase retrieval with random initialization (\textbf{GD-Rand}) and spectral initialization~\cite{netrapalli2013phase} (\textbf{GD-Spec}), a plug-and-play recovery method (\textbf{PnP-FISTA})~\cite{kamilov2017plug} that uses a pretrained \textbf{DRUNet} denoiser~\cite{zhang2021plug} as a denoiser prior, the wavefront sensing method \textbf{WISH}~\cite{wu2019wish} that uses Gerchberg–Saxton iterations~\cite{gerchberg1972practical, fienup1982phase}, and a Deep Image Prior (\textbf{DIP})–based approach~\cite{ulyanov2018deep} using a  convolutional decoder architecture from~\cite{darestani2021accelerated} as a phase prior. Although the original WISH uses random phase modulation, we replace it with phase ramps to generate shifted measurements in our setup. Additional details are provided in the supplementary material.

We evaluate phase recovery using the mean absolute phase error after global phase correction. Given estimated phasors $\widehat{\bmp}$ and ground-truth phasors $\bmp_{gt}$, we estimate the global phase offset $\theta$ from the circular mean of their phase differences and measure the residual error as
$\mathcal{E}_{\text{phase}} = \frac{1}{N} \sum_{i=1}^{N} 
\Big|\angle\big(\widehat\bmp(i) \bmp^*_{gt}(i) e^{-j \theta}\big)\Big|.$
We report the mean phase error ($\pm$ standard deviation) at SNR $\approx 22$ dB in~\cref{tab:phase_error_snr_22}, where our method achieves the lowest error for both quadratic and smooth phase profiles. As shown in~\cref{fig:viscomp_snr22_peak}, with 16 measurements it is the only approach that reconstructs a phase profile consistent with the ground truth. For random phase, gradient descent with random initialization yields a lower error (see~\cref{fig:viscomp_snr22_peak}). The performance gap is small and disappears at higher SNR, as shown in~\cref{fig:comp_methods}. \Cref{fig:comp_methods} further shows that least-squares refinement (\textbf{Ours + LS}) yields clear improvements under noisy measurements. Additional results under varying noise and measurement settings are provided in the supplementary material.

\noindent
\textbf{Computation time.} We compare runtime per reconstruction across methods on a single NVIDIA RTX 6000 Ada GPU. Our approach is the fastest at $1.2$ s, compared to $1.5$ s for gradient-based optimization, $3.6$ s for WISH, $8.7$ s for DIP, and $10.7$ s for PnP-FISTA. Overall, this corresponds to an approximately $8\times$ reduction in computation time relative to the slowest baseline.

\begin{table}[t]
\centering
\small
\setlength{\tabcolsep}{3pt}
\caption{Phase errors (mean $\pm$ std) at 22\,dB SNR. Our method achieves the best or second-best results across all phase profiles. (Best bold, second-best underlined).}
\label{tab:phase_error_snr_22}
\resizebox{0.95\textwidth}{!}{%
\begin{tabular}{lcccccc}
\toprule
& \multicolumn{2}{c}{Quadratic phase} & \multicolumn{2}{c}{Random phase} & \multicolumn{2}{c}{Smooth phase (peaks) } \\
\cmidrule(lr){2-3} \cmidrule(lr){4-5} \cmidrule(lr){6-7}
\diagbox{Method}{\#Meas.} & 16 & 32 & 16 &  32 & 16 &  32 \\

\midrule

GD-Rand & 0.695 $\pm$ 0.365 & 0.470 $\pm$ 0.434 & \textbf{0.138 $\pm$ 0.148} & \textbf{0.079 $\pm$ 0.152} & 0.791 $\pm$ 0.361 & 0.570 $\pm$ 0.404 \\
GD-Spec & 0.737 $\pm$ 0.400 & 0.400 $\pm$ 0.432 & 0.691 $\pm$ 0.367 & 0.449 $\pm$ 0.440 & 0.686 $\pm$ 0.397 & 0.371 $\pm$ 0.374 \\
WISH & 0.852 $\pm$ 0.299 & 0.423 $\pm$ 0.366 & 0.854 $\pm$ 0.320 & 0.453 $\pm$ 0.377 & 0.898 $\pm$ 0.330 & 0.425 $\pm$ 0.367 \\
PNP-FISTA & 0.669 $\pm$ 0.341 & 0.456 $\pm$ 0.304 & 0.735 $\pm$ 0.329 & 0.425 $\pm$ 0.264 & 0.663 $\pm$ 0.341 & 0.479 $\pm$ 0.338 \\
DIP & 1.486 $\pm$ 0.031 & 1.491 $\pm$ 0.051 & 1.495 $\pm$ 0.037 & 1.447 $\pm$ 0.065 & 1.442 $\pm$ 0.063 & 1.423 $\pm$ 0.081 \\ [0.5ex]

\hdashline \noalign{\vskip 0.5ex}
Ours & \textbf{0.158 $\pm$ 0.104} & \textbf{0.099 $\pm$ 0.083} & \underline{0.156 $\pm$ 0.099} & \underline{0.094 $\pm$ 0.075} & \textbf{0.155 $\pm$ 0.096} & \textbf{0.097 $\pm$ 0.071} \\
Ours with LS & \underline{0.183 $\pm$ 0.082} & \underline{0.126 $\pm$ 0.065} & 0.158 $\pm$ 0.075 & 0.121 $\pm$ 0.055 & \underline{0.205 $\pm$ 0.063} & \underline{0.159 $\pm$ 0.050} \\
\bottomrule
\end{tabular}
}

\end{table}

\begin{figure}[t] %
    \centering
    \includegraphics[width=\linewidth]{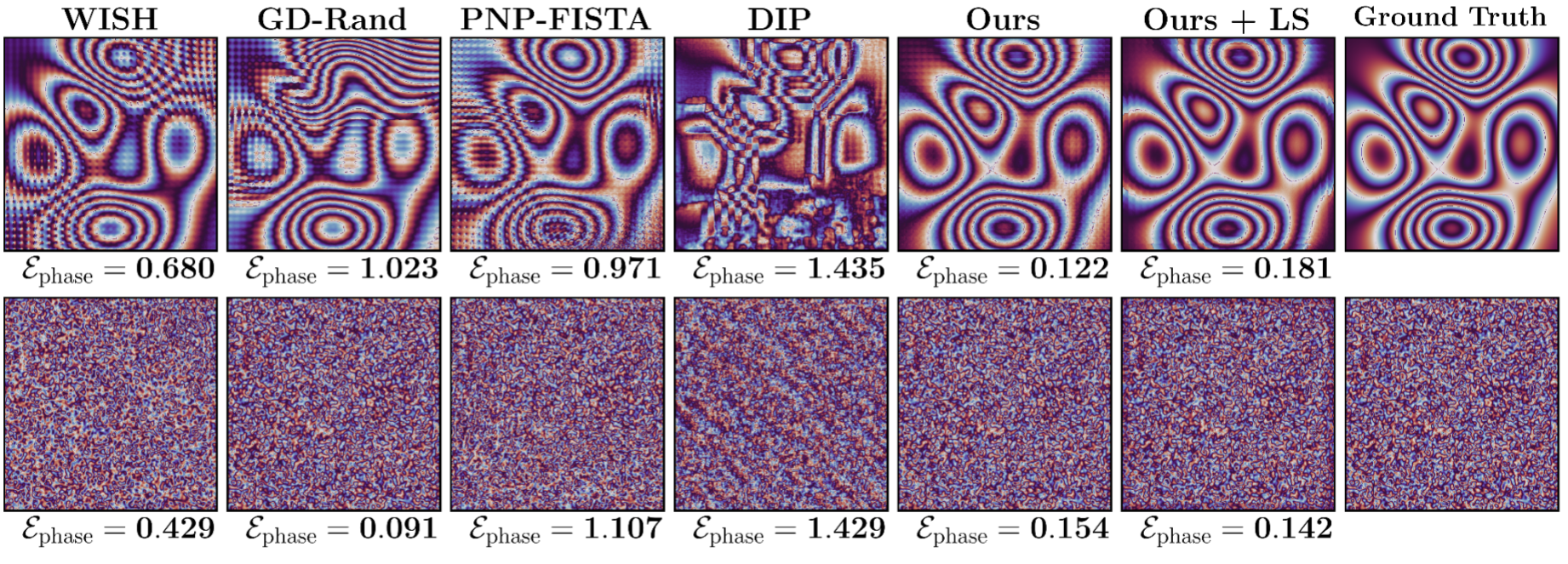}
    \caption{\textbf{Phase recovery of smooth and random phase profiles} (SNR 22 dB, $16$ measurements). Our method achieves the lowest error on smooth profiles and remains competitive for random phases.}

    \label{fig:viscomp_snr22_peak}
\end{figure}

\begin{figure}[t]
    \centering
    \includegraphics[width=\linewidth]{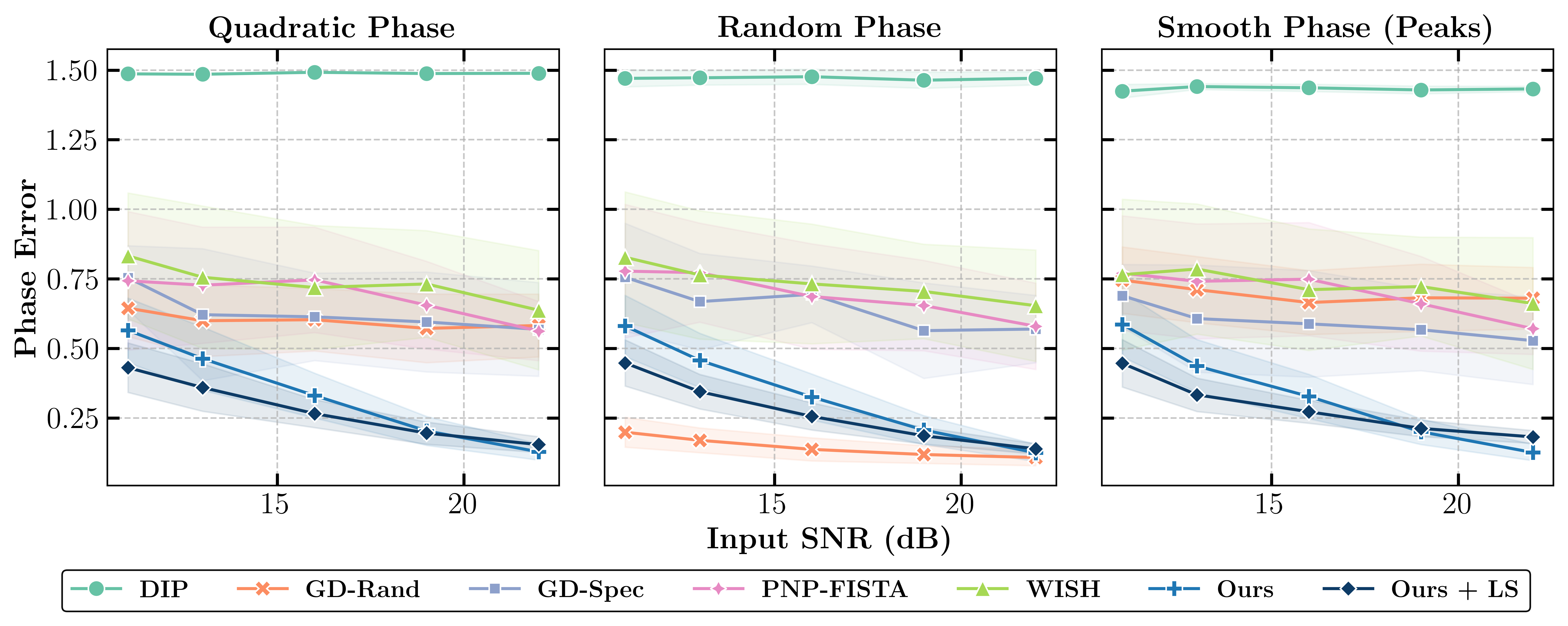}
    \caption{\textbf{Phase recovery performance across SNRs in simulations.} Our method demonstrates consistent performance across all phase types, maintaining the lowest error in most settings and highly competitive results otherwise. Notably, least-squares refinement (\textit{Ours\,+\,LS}) significantly improves reconstruction accuracy at low SNRs.}
    \label{fig:comp_methods}
\end{figure}

\newpage 
\subsection{Hardware prototype and real experiments} 
\begin{figure}[t]
    \centering
       \includegraphics[width=0.6\linewidth]{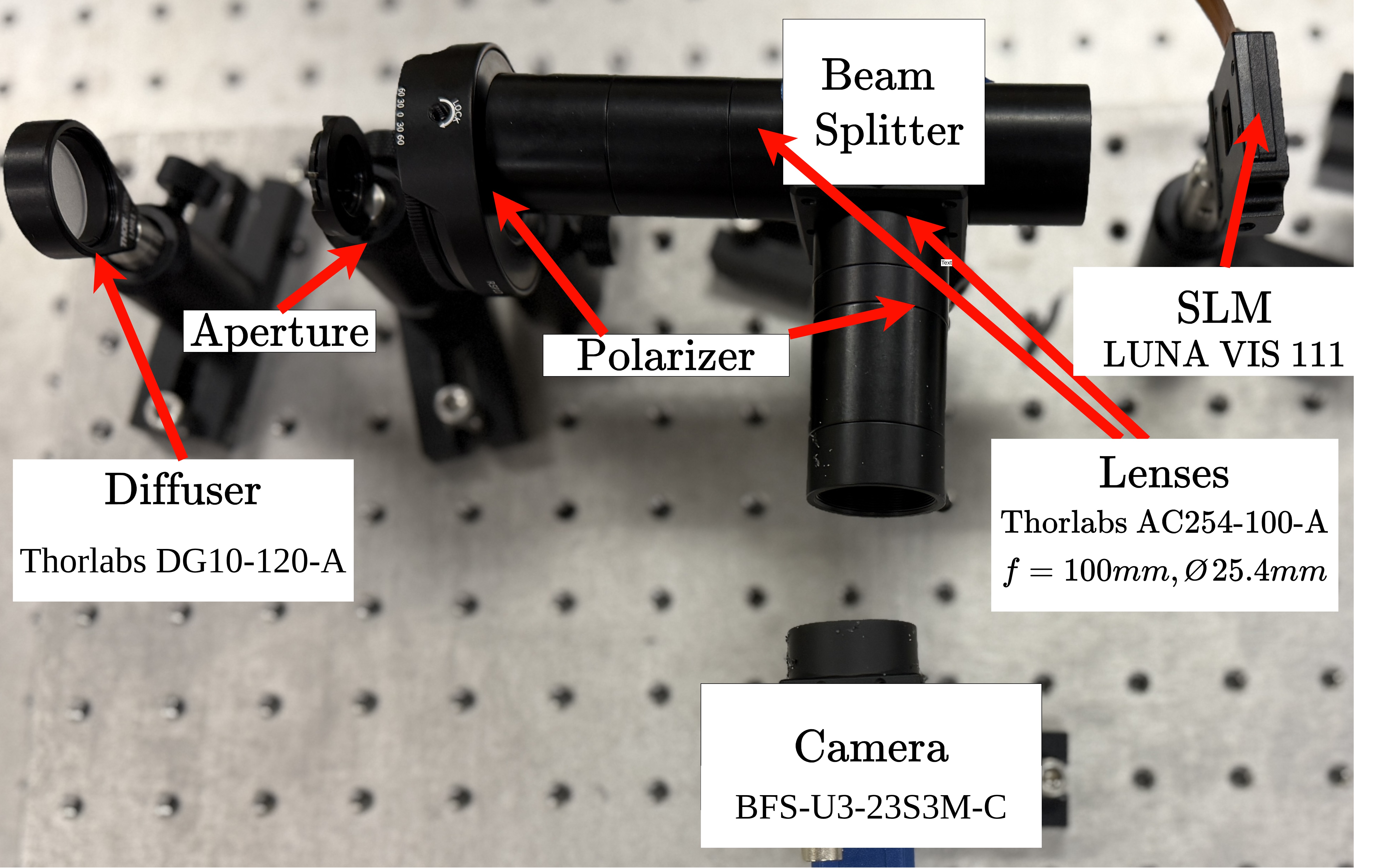}
    \caption{Hardware prototype.}
    \label{fig:hardware_prototype}
\end{figure}
We implement a hardware prototype shown in \cref{fig:hardware_prototype}, following the optical layout in Fig.\ \ref{fig:intro_figure}, to further validate the proposed self-referenced PSI framework. A $4f$ optical system with a phase-only SLM at the Fourier plane generates spatial shifts via phase ramps. This setup allows us to record both interference patterns for phase retrieval and unmodulated captures of the scene to directly measure the amplitude. 
Additional hardware implementation details are provided in the supplementary material.

\noindent
\textbf{Optical phase recovery.} We place objects in the input plane and recover their phase profile. \Cref{fig:real_exp_lens} shows the recovered phase profile of a convex lens placed in the object plane under different shifts. We observe that the optimal shift pairs $(24, 25)$ for the setup produce the most accurate and stable reconstructions. Other shift pairs, as shown in~\cref{fig:real_exp_lens}, lead to suboptimal phase estimates. 

\begin{figure}[t]
    \centering
    \includegraphics[width=0.9\linewidth]{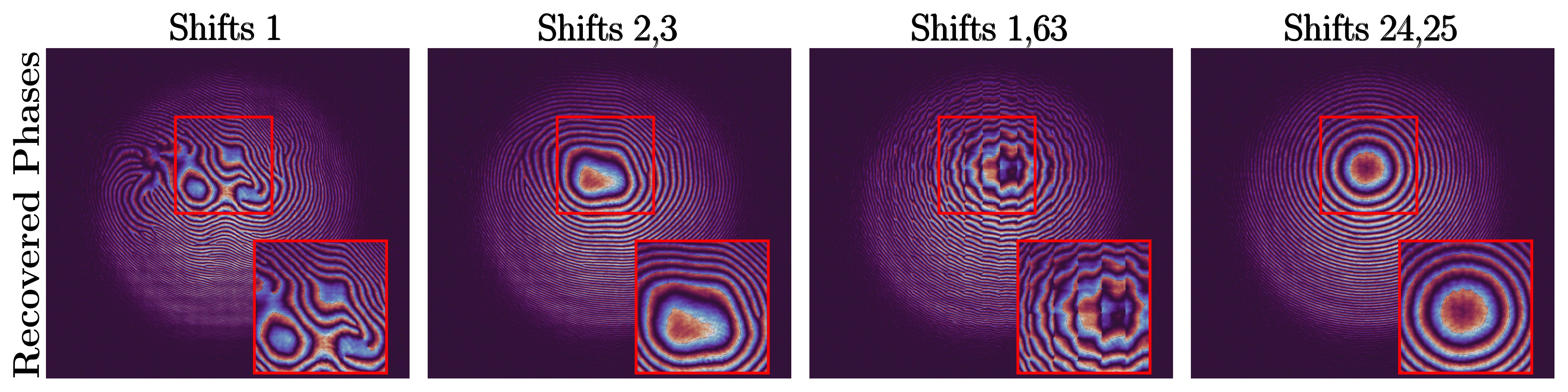}
   \caption{\textbf{Optical profile recovery of a convex lens in real experiments.} Optimal coprime shifts $(24, 25)$ provide a smooth and accurate phase recover of the thin lens. In contrast, other non-optimal shifts introduce phase discontinuities and artifacts.}
    \label{fig:real_exp_lens}
\end{figure}

\noindent
\textbf{Target refocusing.} We place a target away from the input plane and capture it out-of-focus. Our algorithm recovers the complex field at the input plane, whose amplitude remains blurred (first column of \cref{fig:target_focusing_star}). Applying Fresnel propagation~\cite{goodman2005introduction} to the recovered wavefront, the target progressively refocuses with increasing propagation distance, yielding a sharp image at the correct depth (last column of \cref{fig:target_focusing_star}). This confirms the recovered phase is physically consistent and preserves the correct wavefront profile required for accurate refocusing. Additional refocusing examples are provided in the supplementary material.

\begin{figure}[bt]
    \centering
    \includegraphics[width=\textwidth]{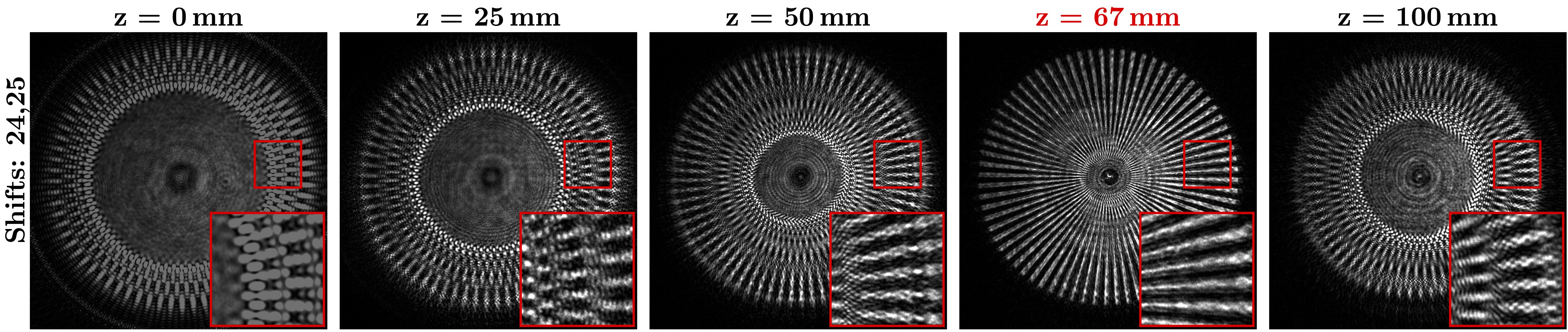}
    \caption{\textbf{Target refocusing in real experiments.} We recover an out-of-focus the complex field and numerically propagate it to multiple distances. Optimal focus is achieved at $z = 67$ mm, 
    where the star pattern and fine features are clearly resolved.}
    \label{fig:target_focusing_star}
\end{figure}

\noindent
\textbf{Seeing through a diffuser.} We place a thin diffuser (DG10-120-A, Thorlabs) in front of the target, which strongly scatters the incident wavefront and produces a speckle-like intensity pattern (first column of~\cref{fig:seeing_through_diff}). Our method reconstructs the complex field at the input plane (second column), which is then numerically propagated to the diffuser plane. There, we undo the effects of the diffuser by multiplying with the conjugate of the calibrated diffuser phase profile. The phase maps after correction (third column of~\cref{fig:seeing_through_diff}) show the structure of the target. Propagating the corrected field to the target plane reveals the hidden object in both amplitude and phase (last column).

\begin{figure}[ht]
    \centering
    \includegraphics[width=\linewidth]{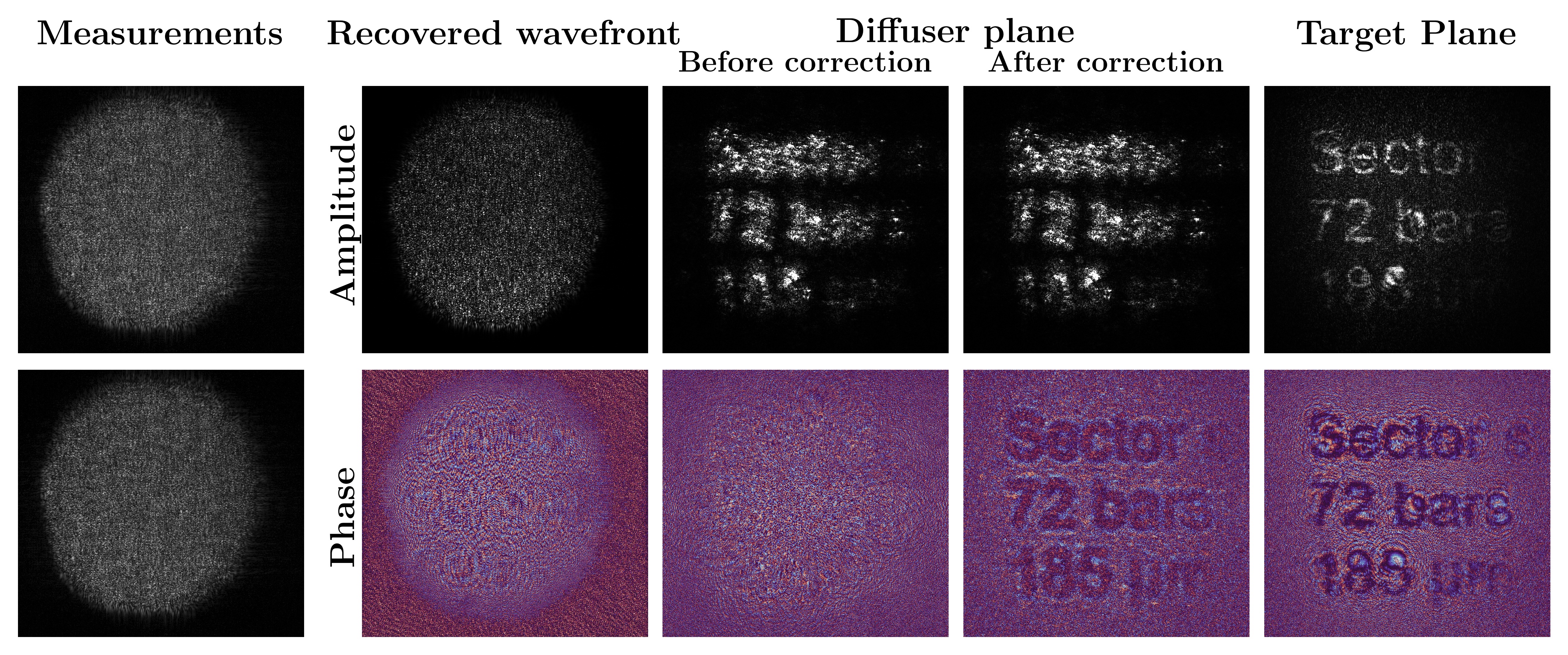}
   \caption{\textbf{Seeing through a diffuser in real experiments.} We recover the scattered complex field and computationally propagate it to the diffuser plane for phase correction. Subsequent propagation to the object plane successfully reveals the target, which is otherwise invisible.}
    \label{fig:seeing_through_diff}
\end{figure}

\section{Conclusion}
We present a principled framework for wavefront phase recovery based on shifted self-interference measurements. Our method provides a simple and analytical solution for phase retrieval by propagating phase differences on a connected graph. The choice of measurement shifts controls the graph structure and error accumulation. We proved that co-prime shift pairs guarantee full graph connectivity and derived a provably optimal shift design that achieves the minimum hop length and worst-case error. Simulation experiments confirm that our optimal shift design consistently outperforms several existing methods. Real hardware experiments validate our approach on practical applications. These results demonstrate that our theoretically grounded measurement design and recovery algorithm provide an accurate and robust wavefront sensing method.

\section*{Acknowledgements}
This paper is partially based on work supported by the NSF CAREER award CCF-2046293.

\clearpage
\bibliographystyle{splncs04}
\bibliography{main}

\begin{center}
{\LARGE \bfseries Supplementary Material}
\end{center}

\newcommand{\beginsupplement}{%
        \setcounter{table}{0}
        \renewcommand{\thetable}{A\arabic{table}}%
        \setcounter{figure}{0}
        \renewcommand{\thefigure}{A\arabic{figure}}%
        \setcounter{section}{0}
        \setcounter{lemma}{0}
        \setcounter{theorem}{0}
        \renewcommand{\thesection}{A\arabic{section}}%
}

\beginsupplement

This supplementary material provides theoretical proofs, implementation details, experimental settings, and additional results supporting the main paper. Complete proofs of all lemmas and theorems are presented in~\cref{supp:sec:proofs}, along with additional noise analysis in~\cref{supp:sec:noise_analysis}. We discuss additional algorithm implementation details in~\cref{supp:sec:algo_details}. We provide hardware specifications, experimental settings, and additional real and simulated experiments in~\cref{supp:sec:hardware,supp:sec:comp_details,supp:sec:sim_exp}.

\section{Proofs of Theoretical Results}
\label{supp:sec:proofs}
In this section, we provide complete proofs of our lemmas and theorems stated in the main paper. For completeness, we first restate the relevant notation. Each pixel or measurement location $i \in \{-\lfloor {N}/{2}\rfloor, \dots, \lfloor {N}/{2}\rfloor-1\}$ corresponds to a \emph{node} in the graph. The set of graph nodes is
\begin{equation}
V = \left\{-\left\lfloor\frac{N}{2}\right\rfloor, \dots, \left\lfloor\frac{N}{2}\right\rfloor - 1\right\},
\label{supp:eq:graph_nodes}
\end{equation}
where $|V| = N$. 
An \emph{edge} exists between node $i$ and $j$ if we have direct phase difference measurements between them as $\nabla_k \bmphi$. For a set of shifts $\mathcal{S} = \{s_1, s_2, \dots, s_K\}$, the edge set is
\begin{equation}
    \label{supp:eq:graph_edges}
    E(\mathcal{S}) = \bigl\{ (i, i \pm s_k) : i \pm s_k \in V,\; s_k \in \mathcal{S} \bigr\}.
\end{equation}
In other words, each node has two neighbors for each shift. We call one edge traversal as a step or hop and the total number of hops/steps required to go from a reference node to any other node as the hop distance.

\subsection{Connectivity guarantees}

We now provide a complete proof of our~\cref{supp:theorem:connectivity}. Our main objective is to show that any node in the graph can be accessed from the reference node through a sequence of hops/steps along the connected nodes. In other words, given two co-prime integers $s,t$, we can represent any node as $n = \alpha s + \beta t$, using integers $\alpha,\beta$ that keep the entire path within the graph $(V,E)$.  

Our required condition resembles the well-known Bezout's identity~\cite{Bezout1779theorie}, which states that any integer can be written as a linear combination of two co-primes $s$ and $t$. This identity alone does not prove connectivity in our case because the intermediate nodes visited by such combinations are not guaranteed to remain within the valid range of $V$. In other words, the identity only states the existence of integers $\alpha,\beta$ such that $\alpha s + \beta t = n$, but the corresponding sequence of hops (\ie $\alpha s$ or $\beta t$) may leave the valid domain before reaching the desired node. Therefore, we need some additional constraints on $s$ and $t$ to prove that valid propagation paths exist entirely inside $V$ that connect all nodes.

Our proof proceeds in two steps. First, \cref{supp:lemma:co-prime_cover_all_residues} shows that when $s$ and $t$ are co-prime, repeated hops of size $s$ generate all possible residues modulo $s+t$ (\ie $0,1,\dots,s+t-1$). Note that a $-t$ step is equivalent to a $+s$ step modulo $s+t$ since $-t \equiv s \pmod{s+t}$. Thus, propagation using $+s$ hops, with $-t$ hops when the traversal reaches the boundary of the interval, visits all residues modulo $s+t$. As illustrated in \cref{supp:fig:connectivity}, these hops cover every index within any interval of length $s+t$. Next, \cref{supp:lemma:sliding_window} shows that such intervals can be translated across the domain while remaining inside $V$ as long as $s+t \leq N$. This allows connectivity to propagate from the reference node to the entire graph. Combining these results establishes the global reachability and connectivity stated in~\cref{supp:theorem:connectivity}.

\begin{figure}[ht]
    \centering
    \includegraphics[width=\linewidth]{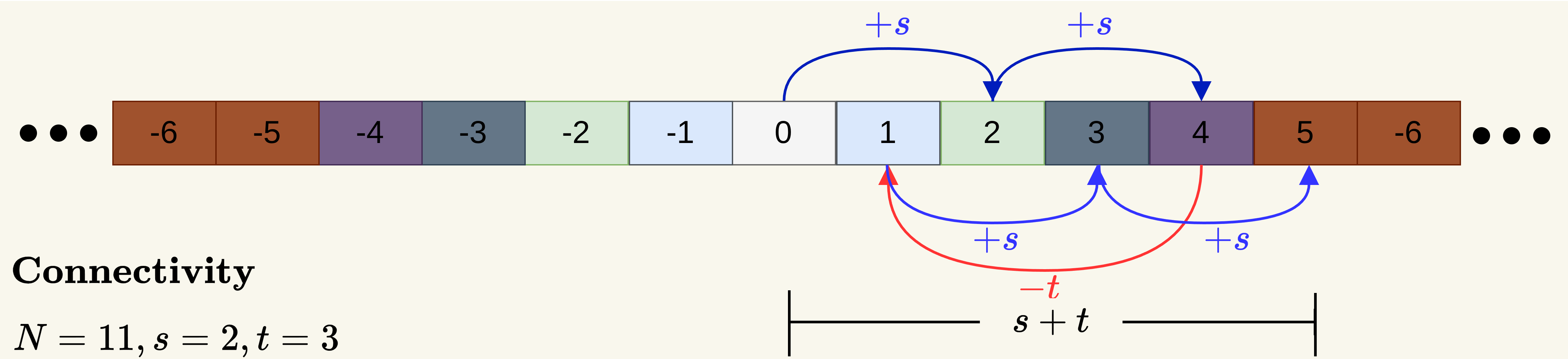}
   \caption{Illustration of the connectivity using two co-prime shifts for $N=11$, $s=2$, and $t=3$ using hops  $+s$ and $-t$. 
    Blue arrows represent forward propagation using $+s$, while the red arrow represents taking a step of size $-t$ when the boundary is reached. 
      Over an interval of length $s+t$ since $s$ and $t$ are co-prime, the hops will visit all nodes without repetition.}
    \label{supp:fig:connectivity}
\end{figure}

\label{supp:sec:connectivity}

\begin{lemma}[Co-prime Residue Coverage]
\label{supp:lemma:co-prime_cover_all_residues}
Let $k\ge 1$, $s\in\mathbb{Z}$ and $\gcd(s,k)=1$. Then for any $p\in\mathbb{Z}$,
\[
\{ (p + is)\bmod k \;:\; i=0,1,\dots,k-1 \,\}=\{0,1,\dots,k-1\}.
\]
\end{lemma}

\begin{proof}
Consider the $k$ residues $\{ (p + is)\bmod k \;:\; i=0,1,\dots,k-1 \,\}$. Showing that each of the $k$ residues is unique will complete the proof, because we know there are $k$ possible non-negative remainders that are less than $k$. If they are distinct, it means that the set is exactly $\{0,1,\dots,k-1\}$. 

Let us assume that they are not distinct and for some $i_1,i_2$ where $0\le i_1<i_2\le k-1$ and $i_1 \neq i_2$ we have
\begin{equation}
    \label{supp:eq:lemma_1_cong}
    p + i_1s \equiv p + i_2s \pmod{k},
\end{equation}
which implies $(i_2-i_1)s \equiv 0 \pmod{k}$. This is true because \cref{supp:eq:lemma_1_cong} would imply that $\exists\,m,n \in \mathbb{N}$ such that 
\[
(p + i_1s)-nk = (p + i_2s) - mk \implies (i_2-i_1)s = (m-n)k.
\] 
We also know that $\gcd(s,k)=1$, $(i_2-i_1)s \equiv 0 \pmod{k}$, which leads to $i_2-i_1 \equiv 0 \pmod{k}$. This is a contradiction because we know $1\le i_2-i_1 \le k-1$ and $i_2-i_1 \equiv 0 \pmod{k}$ is valid if and only if $i_2 = i_1$, which contradicts our assumption.

The argument above proves that all $k$ residues are distinct and the set must exactly be $\{0,\dots, k-1\}$.  
\end{proof}

\begin{lemma}[Sliding Window Coverage]
\label{supp:lemma:sliding_window}
Let $s, t \in \mathbb{N}$ with $s + t \le N$, and let $V$ be defined as~\cref{supp:eq:graph_nodes}. For any $x \in V$, there exists an interval $I_1$ of size $s+t$ containing $x$, and a sequence of intervals $I_1, I_2, \dots, I_K$ of size $s+t$, each contained in $V$, such that consecutive intervals share at least one common element and $I_K$ contains $0$.
\end{lemma}

\begin{proof}
Since $s+t \le N = |V|$, any interval of size $s+t$ fits within $V$. For any $x \in V$, set the starting point of $I_1$ as 
\begin{equation}    
 I_1 = [i_{0}, i_{0} +s + t - 1] \quad
\text{ where } \quad 
i_{0} = \min(x, \left\lfloor\frac{N}{2}\right\rfloor - s -t + 1),
\end{equation}
which contains $x$ and lies within $V$. We then construct $I_2, \dots, I_K$ by shifting each interval by one integer toward $0$, so the consecutive intervals overlap in $s+t-1$ elements. This sequence ends in at most $|i_0|$ steps, as each shift reduces the distance to $0$ by exactly one. Thus, we guarantee $I_K$ contains $0$.%
\end{proof}

\begin{theorem}[Connectivity via Co-Prime Shifts]
\label{supp:theorem:connectivity}
Let $s,t \in \mathbb{N}$ be co-prime, i.e., $\gcd(s,t)=1$, and let the graph $V$ be defined as~\cref{supp:eq:graph_nodes} and edges $E\left(\{s,t\}\right)$ defined following~\cref{supp:eq:graph_edges} using shifts $s$ and $t$. 
Then every node in $V$ is reachable from the reference node $0$ through a sequence of hops along edges, provided
$
s + t \le N.
$
\end{theorem}

\begin{proof}
    By Bezout's identity\cite{Bezout1779theorie} since $s$ and $t$ are co-prime, for any integer $n$, there exist integers $\alpha, \beta$ such that
    \[
    \alpha s + \beta t = n.
    \]
    Therefore, any integer $n$ can be expressed as a linear combination of $s$ and $t$. In particular, we can take $\alpha$ hops of size $s$ and $\beta$ hops of size $t$ to reach $n$. However, this identity alone does not establish connectivity, as it does not guarantee that the hops land on valid nodes inside $V$.

     Next, we show that under the condition $s+t \leq N$, there exists a valid sequence of paths that remain entirely within $V$ and connect any node in $V$ to the reference node. Consider the propagation sequences $n_k$, where we start from node $n_0 = 0$ take $+s$ steps to the right until we reach the boundary $\lfloor {N}/{2}\rfloor$. If we reach the boundary, we take the step $-t$. We can write these sequences as 
    \begin{equation}
     n_k = \begin{cases}
         n_{k-1} + s \quad \text{ if } n_{k-1} + s \leq \lfloor {N}/{2}\rfloor \\
         n_{k-1} - t \quad \text{ otherwise }.\\
     \end{cases}
    \end{equation}
    Notice that the iteration can proceed for arbitrarily many $k$ without leaving $V$, i.e., $n_k \in V$ for all $k$. This is because if $n_{k-1} + s > \frac{N}{2}$, then we must have $n_{k-1} - t > -\frac{N}{2}$. This follows from
    \[
    \text{if } n_{k-1} + s > \frac{N}{2} \implies n_{k-1} - t > \frac{N}{2} - s - t
     \ge 
    \frac{N}{2} - N
     \ge 
    -\frac{N}{2}.
    \]
    
    We can see that $n_k$ and $n_0 + ks$ have the same remainder when divided by $s+t$. Formally, $n_k \equiv n_0 + ks \pmod{s+t}$.  
    At each step, the update is either
    \[
    n_{k+1}=n_k+s
    \quad\text{or}\quad
    n_{k+1}=n_k-t.
    \]
    In both cases,
    $
    n_{k+1}\equiv n_k+s \pmod{s+t},
    $
    since $-t \equiv s \pmod{s+t}$. Therefore, 
    $
    n_k \equiv n_0 + ks  \pmod{s+t}.
    $
    
    Let us now consider an interval of length $s+t$ starting from node $0$. 
    Since $\gcd(s,t)=1$, it follows that $\gcd(s,s+t)=1$. 
    By~\cref{supp:lemma:co-prime_cover_all_residues}, the sequence $n_0 + ks$ for $k=0,\dots,s+t-1$ attains all residues $\{0,1,\dots,s+t-1\} \mod s+t$. Since the hop sequence $n_k$ is congruent with $n_0 + ks \pmod{s+t}$, the sequence $n_k$ connects all nodes starting from $n_0$ up to $n_{s+t-1}$ (\ie the sequence forms a contiguous interval of length $s+t$.).
    
    Finally, it remains to show that for any $n \in V$, there exists a valid path from the reference node to $n$ using $+s$ and $-t$ hops. By~\cref{supp:lemma:sliding_window}, for any $n \in V$ there exists a sequence of intervals $I_1, \dots, I_K$ each of size $s+t$ and contained in $V$, with $n \in I_1$, consecutive intervals sharing a common node, and $0 \in I_K$. Since every node within each interval is mutually reachable via $+s$ and $-t$ steps, the shared nodes between consecutive intervals provide a valid path from $n$ to $0$.%
\end{proof}

\subsection{Optimal co-prime shift design}

In this section, we provide the complete proof of our theorem for the optimal shift pairs. We first establish a lower bound on the maximum hop distance $h^\star$ that is necessary for full connectivity by any pair of shifts (see~\cref{supp:lemma:min_hops}). This lemma counts the maximum number of nodes that can be reached at each hop distance. Note that each node has four possible neighbors with two shifts (i.e., separated by $\pm s$ and $\pm t$). By counting the number of nodes reachable within $h$ hops, we obtain an upper bound on the number of distinct nodes that can be visited after $h$ hops (or steps). 
The node count in turn can help us determine how many hops are required to cover all $N$ nodes. 
Since the node count ignores the fact that multiple paths may lead to the same node, the count gives us a lower bound on the number of hops required for full connectivity. In summary, the count represents the maximum possible number of reachable nodes and provides a universal lower bound on the required hop length.

\begin{lemma}[Lower Bound on the Maximum Hop Distance]
\label{supp:lemma:min_hops}
    Consider the graph $V$ with $N$ nodes defined in~\cref{supp:eq:graph_nodes} and edges $E(\{s,t\})$ as in~\cref{supp:eq:graph_edges}. Let $h^\star$ denote the minimum number of hops required to reach all nodes from a reference node. Then, for any choice of two shifts,
    \begin{equation}
    \label{supp:eq:min_hop_bound}
    h^\star \ge \left\lceil \tfrac{-1 + \sqrt{2N - 1}}{2} \right\rceil.
    \end{equation}
\end{lemma}

\begin{proof}  If we fix a hop budget $h$, any node reachable from the reference within $h$ hops must be 
written as
\begin{equation}
    \label{supp:eq:reachable_nodes}
    x = \alpha s + \beta t,
    \qquad
    |\alpha| + |\beta| \le h,
\end{equation}
which means we take $\alpha$ hops using shift $s$ and $\beta$ hops using shift $t$. Therefore, the set of nodes reachable within $h$ hops is
\[
\mathcal{S}_h = 
\{ \alpha s + \beta t : |\alpha| + |\beta| \le h \}.
\] 
The maximum possible size of this set can be 
\[
1 + 4 \sum_{i=1}^{h} i
=
2h^2 + 2h + 1,
\] which is found by counting the integer points in the diamond $|\alpha| + |\beta| = i, \text { for } i = 1, \dots , h$. Since different pairs $(\alpha,\beta)$ may produce the same node $n$, the number of 
distinct nodes reachable within $h$ hops satisfies
$
|\mathcal{S}_h|
\le
2h^2 + 2h + 1.
$
In order to reach all $N$ nodes in the interval, it is therefore necessary that
\[
2h^2 + 2h + 1 \ge N.
\]
Solving this quadratic inequality gives
\[
h \ge \frac{-1 + \sqrt{2N - 1}}{2} .
\] Since $h$ is an integer for any $N>1$, we obtain the universal lower bound
\[ 
{h}^\star = \left\lceil
\frac{-1 + \sqrt{2N - 1}}{2} 
\right\rceil.
\] 
This result implies that, for any pair of two shifts $s$ and $t$, we will not  be able to connect all the nodes in $V$ using less than ${h}^\star$ hops (or steps). %
\end{proof}

Next, we prove~\cref{supp:theorem:optimal_coprime_pair}, which shows that the co-prime shifts 
$s^\star = \lfloor \sqrt{N/2} \rfloor$ and 
$t^\star = s^\star  + 1$ 
achieve this bound. 
We first show that using the lower bound $h^\star$ hops, the nodes visited by the optimal shifts form a continuous interval. 
Finally, we prove that the endpoints of this interval exceed $\left\lfloor N/2 \right\rfloor$ in magnitude, which implies that the entire node set $V$ lies within the reachable region. Note that this theorem does not claim that the provided optimal pair is unique. Experimentally, we observe that some other shift pairs with values close to $\sqrt{N/2}$ also achieve the same bound.

\begin{theorem}
    \label{supp:theorem:optimal_coprime_pair}
    For a graph with $N$ nodes as in~\cref{supp:eq:graph_nodes} and edges defined by two shifts $s,t\in\mathbb{N}$, 
    $s = \left\lfloor \sqrt{\frac{N}{2}} \right\rfloor$ and 
    $t = \left\lfloor \sqrt{\frac{N}{2}} \right\rfloor + 1$
    achieve the minimum possible hop distance $h^\star$ given in~\cref{supp:lemma:min_hops}.
\end{theorem}

\begin{proof}
    
Let us define consecutive co-prime shifts as $s = k$ and $t = k+1$ for a $k \in \mathbb{N}$. Since consecutive integers satisfy $\gcd(s,t)=1$, this choice of $s$ and $t$ will give us a co-prime set and any integer reachable within hop budget $h$ can be written as
\begin{align}
     \label{supp:eq:reachable_tho2}
    n &= \alpha k + \beta (k+1),
    \qquad |\alpha|+|\beta|\le h \\
    &= (\alpha + \beta )k + \beta. \notag
\end{align}
Let us define an anchor point $u=\alpha + \beta$ and offset $v=\beta$ and rewrite~\cref{supp:eq:reachable_tho2} as
\begin{equation}
     n = uk + v. \\
\end{equation}
We can view the propagation as making $u$ hops of $k$ units and visiting $v$ consecutive nodes. For a fixed budget $h$, the constraint on $u$ and $v$ can be written as
\[
|u-v| + |v| \le h.
\] By the triangle inequality, we have $|u| = |\alpha + \beta| \leq |\alpha| + |\beta| \leq h$.  Similarly for a fixed $u$, the
corresponding values of $v$ must satisfy $u-h \leq 2v \leq u + h  $, which can be written as 
\[
\left\lceil \frac{u-h}{2} \right\rceil
\le
v
\le
\left\lfloor \frac{u+h}{2} \right\rfloor.
\]
Therefore, the integers reachable within $h$ hops can be written as
\begin{equation}
n = uk + v,
\qquad
|u|\le h,
\quad
\left\lceil \frac{u-h}{2} \right\rceil
\le
v
\le
\left\lfloor \frac{u+h}{2} \right\rfloor.
\end{equation}
For each fixed $u$, the range of possible values of $v$ consists of consecutive integers between 
$\lceil \frac{u-h}{2} \rceil$ and $\lfloor \frac{u+h}{2} \rfloor$. Thus, we can define the interval of nodes visited by a given $u$ and $k$ as
\[
I_{u,k} = \left[\, uk + \left\lceil \frac{u-h}{2} \right\rceil,\; uk + \left\lfloor \frac{u+h}{2} \right\rfloor \,\right].
\]
We can find that the difference between the end of the interval $I_{u,k}$ and the start of the interval $I_{u+1,k}$ (\ie, the gap between the intervals) is
\begin{equation}
    \label{supp:eq:interval_gap}
\min(I_{u+1, k}) - \max(I_{u,k}) \le k-h.
\end{equation}
Therefore, if $h \ge k$, the intervals $I_{u,k}$ and $I_{u+1,k}$ are either overlapping or contiguous. 

The interval $I_{u,k}$ will have a maximum value of $h(k+1)$ for $u=h$ and a minimum value of $-h(k+1)$ for $u=-h$. Therefore, following~\cref{supp:eq:interval_gap} if $k\leq h$, the union of intervals for $|u|\le h$ is 
\[
\bigcup_{|u|\le h} I_{u,k} = [-h(k+1),\, h(k+1)].
\]

Let us now consider the choices: $k^\star=\left\lfloor \sqrt{\frac{N}{2}} \right\rfloor$, $s=k^\star$, $t=k^\star+1$, 
and $h=h^\star$ from~\cref{supp:lemma:min_hops}. 
First, notice that $h^\star \ge k^\star$ because 
\[
h^\star \geq
\frac{-1 + \sqrt{2N - 1}}{2}  \underset{(i)}{>}  \frac{-1 + \sqrt{2N} - 1}{2} = \sqrt{\frac{N}{2}} - 1,
\]
where $(i)$ comes from the fact that $\sqrt{2N-1} > \sqrt{2N}-1$. Finally, we have 
\begin{equation}
   \label{supp:eq:hstar_geq_kstar}
    h^\star  > \sqrt{\frac{N}{2}} - 1 \implies h^\star \geq \left\lfloor \sqrt{\frac{N}{2}} \right\rfloor = k^\star.
\end{equation}
Second, we need to show that $h^\star(k^\star+1)\ge \lfloor N/2 \rfloor$. The choice of $k=k^\star$ can map to any possible values of $N$ that are in the range $2(k^\star)^2 \leq N < 2(k^\star + 1)^2$. For these values, note that from~\cref{supp:eq:hstar_geq_kstar} $h^\star$ can take the values of $h^\star = k^\star$ or $h^\star = k^\star + 1$. In particular, we have the two cases written as
\[
    h^\star = 
    \begin{cases} 
      k^\star & \text{if }\quad 2(k^\star)^2 \leq N \leq 2(k^\star)^2 + 2k^\star + 1 \\
      k^\star + 1 & \text{if }\quad 2(k^\star)^2 + 2k^\star + 2 \leq N < 2(k^\star + 1)^2. 
    \end{cases}
\]
In both cases, we can verify that 
\begin{equation}
\label{supp:eq:hstar_kstar_cvg}    
h^\star(k^\star+1)\ge \left \lfloor \frac{N}{2} \right\rfloor. 
\end{equation}

We have shown that there are no gaps between the intervals since $h^\star \ge k^\star$, and that we can reach all $N$ nodes of $V$. Consequently, the union of intervals for $|u|\le h^\star$ covers all nodes of $V$. In other words,
\[
V = \left\{-\left\lfloor\frac{N}{2}\right\rfloor, \dots,
\left\lfloor\frac{N}{2}\right\rfloor - 1\right\}
\subset
\bigcup_{|u|\le h^\star} I_{u,k^\star} .
\]
Therefore, the shifts $s=k^\star$ and $t=k^\star+1$ achieve the minimum hop bound $h^\star$ and connect all nodes in $V$. 
\end{proof}

\begin{figure}[!t]
    \centering
    \includegraphics[width=0.9\linewidth]{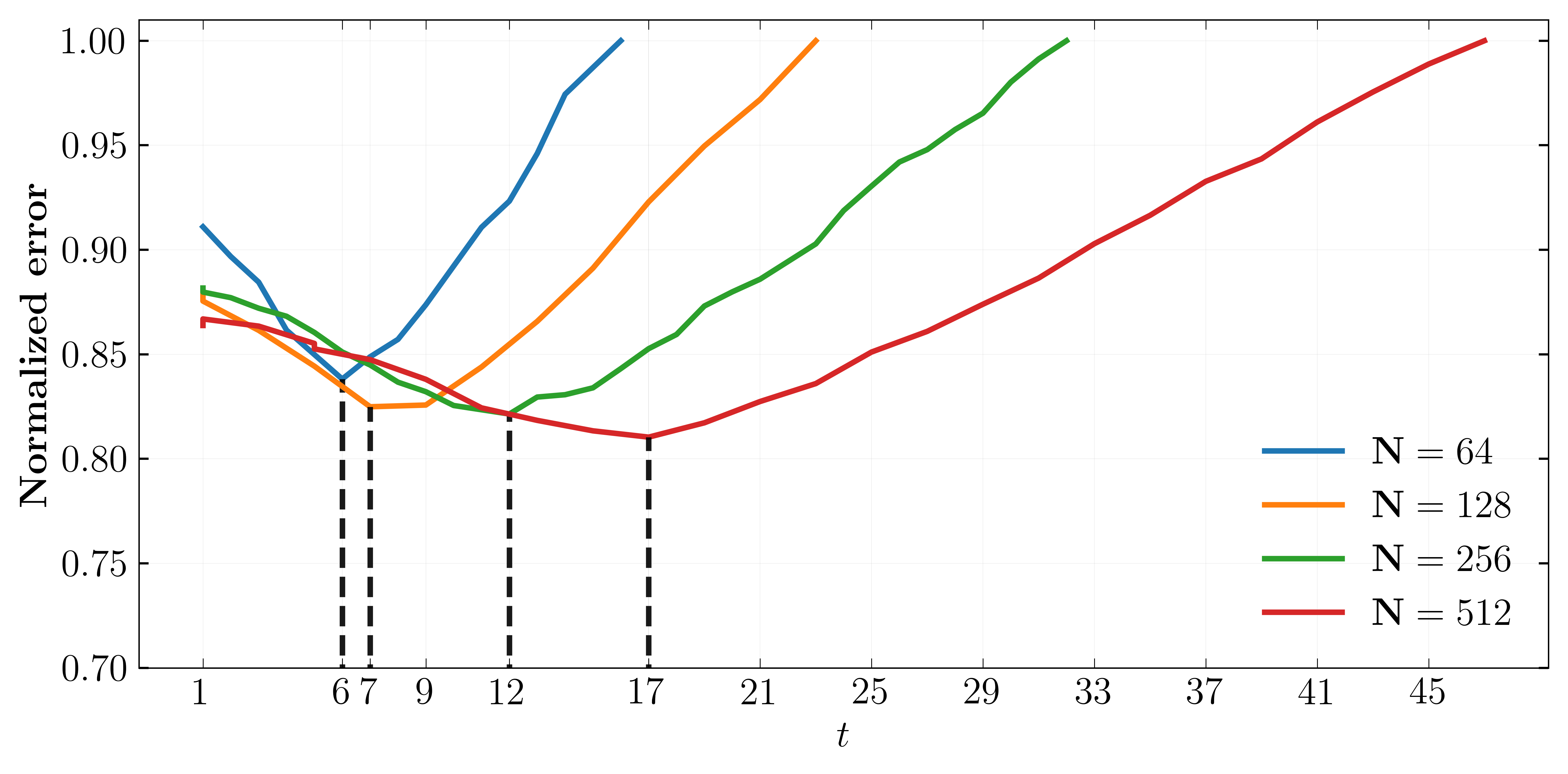}
    \caption{
    Normalized reconstruction error as a function of the shift $t$ for different sizes $N$. 
    For each curve, the first co-prime pair is fixed to $s=k^\star=\lfloor\sqrt{N/2}\rfloor$ and we sweep through all possible choices of $t$.
    Across all tested sizes, the minimum error occurs near $t=s+1$ (indicated using the dashed lines). This confirms the consecutive co-prime pairs $(s,t)=(k^\star,k^\star+1)$ yields the optimal propagation structure and minimizes noise accumulation.
    }
    \label{supp:fig:error_vs_shift_multiple_N}
\end{figure}

\section{Noise Analysis and Optimality}
\label{supp:sec:noise_analysis}

In this section, we provide additional empirical results that support the theoretical analysis. In particular, we investigate how the choice of shifts affects noise propagation and reconstruction error. The \textbf{hop distance} denotes the number of propagation steps required to reach a given node from the reference node. \\

\noindent\textbf{Optimal shift selection.}
To validate~\cref{supp:theorem:optimal_coprime_pair}, we fix the first shift to $s=k^\star$ and sweep the second shift $t$ across a range of values while measuring the reconstruction error. The results for several signal sizes are shown in \cref{supp:fig:error_vs_shift_multiple_N}. The minimum error consistently occurs near $t=s+1$, supporting the theoretical result that the consecutive co-prime pair $(s,t)=(k^\star,k^\star+1)$ achieves the lowest error. 

For the number of nodes $N$ shown in~\cref{supp:fig:error_vs_shift_multiple_N}, we obtain the exact optimal pairs predicted by the theory for most cases: $(s,t)=(5,6)$ for $N=64$, $(11,12)$ for $N=256$, and $(16,17)$ for $N=512$, which follow $k^\star=\lfloor\sqrt{N/2}\rfloor$. For $N=128$, the minimum occurs at $(7,8)$, which is very close to the predicted pair $(8,9)$. The recovery error observed in~\cref{supp:fig:error_vs_shift_multiple_N} is very close, as both shifts achieve similar hop lengths.

\noindent\textbf{Noise reduction with path averaging.}
We can reduce noise accumulation by averaging the phasor estimates obtained from multiple propagation paths. When several paths reach the same node, their estimates can be combined which reduces estimator variance and improves robustness to measurement noise.

\Cref{supp:fig:mae_vs_hop_noavg_vs_avg} compares the cumulative phase error with and without path averaging for several shift configurations. In all multi-shift cases, averaging leads to a clear reduction in reconstruction error across hop distances. This behavior is expected since averaging combines multiple independent estimates of the phase difference and therefore suppresses noise. For the single-shift case there is no improvement since the phase at each node is obtained from only one propagation path. 

\begin{figure}[t]
    \centering
    \includegraphics[width=\linewidth]{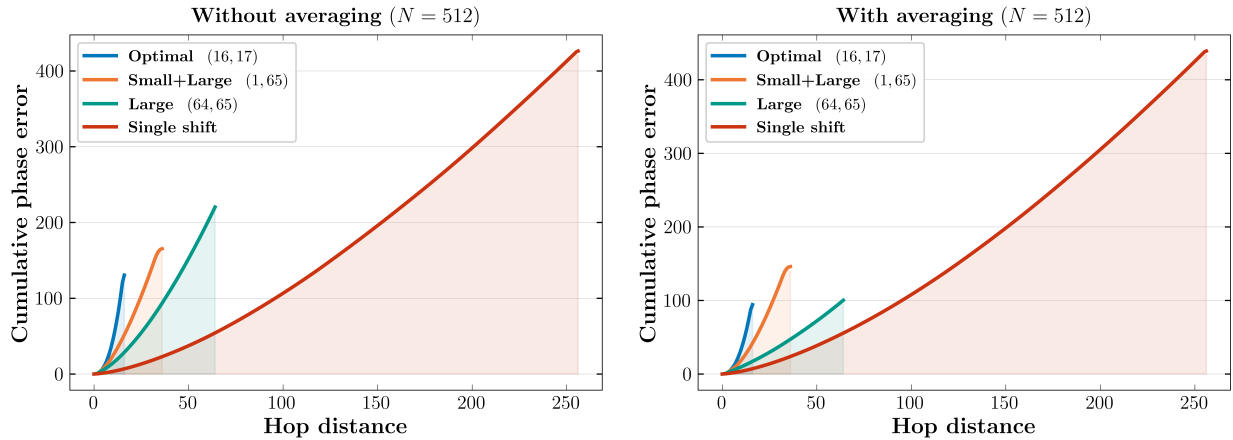}
    \caption{Cumulative phase error as a function of hop distance with and without path averaging for several shift configurations at $N=512$. Averaging reduces noise accumulation when multiple propagation paths exist because estimates from different paths can be combined. 
    }
    \label{supp:fig:mae_vs_hop_noavg_vs_avg}
\end{figure}

\section{Additional Algorithm Implementation Details}
\label{supp:sec:algo_details}

Our main algorithm (Algorithm 1 of the main paper) propagates phase estimates using a breadth–first
search (BFS) starting from a reference pixel. In this formulation, a queue
stores pixels whose phase has already been estimated. Each time a pixel is
dequeued, its neighbors defined by the shifts $\{\Delta_k\}$
are visited. The phases of these neighbors are computed by multiplying the current
phasor estimate with the measured phase differences. 

This queue-based implementation is inherently
sequential and therefore inefficient on GPUs. To address this, we reformulate
the propagation as a parallel wavefront update over the entire signal.
Instead of processing one pixel at a time, we maintain a binary mask that
identifies the currently active pixels whose phases have already been
estimated. At each iteration, all active pixels simultaneously propagate
their phase estimates to neighboring pixels defined by the shifts
$\{\Delta_k\}$. The propagation is implemented using tensor slices that align source and
destination pixels. For each shift $s_k$, we update
\begin{equation}
\widehat{\bmp}_{i+s_k,j}^{(t+1)} \mathrel{+}= 
\widehat{\bmp}_{i,j}^{(t)}\, \bmp^{(v)}_k(i+s_k,j),
\end{equation}
for all active pixels $(i,j)$ simultaneously. Despite propagating from all visited pixels, hop distances are assigned only on first visit. This ensures that each pixel records its shortest-path hop count. After propagation, pixels reached for the first time are assigned their hop distance, and phasors arriving with the same hop length are
averaged. The algorithm updates the active pixels and repeats the same process until all nodes are visited. This parallel formulation preserves the behavior of the BFS
propagation while allowing all updates to be computed efficiently using
batched tensor operations.

\section{Additional Real Experiments and Hardware Details}
\label{supp:sec:hardware}
We designed a hardware prototype to realize a self-referenced PSI model and to validate our method on real experiments. The optical components used in our setup are summarized in~\cref{supp:tab:hardware}. We use a collimated $532$ nm laser diode to illuminate the target object. We built a $4f$ system using two lenses each with focal length $f = 100$ mm. A Holoeye LUNA phase-only SLM ($1920\times1080$, $4.5\,\mu$m pitch) is placed at the Fourier plane of the input. A non-polarizing beam splitter between the first lens and the SLM routes both the unmodulated and modulated reflected wavefronts through the second lens to a camera sensor ($1920\times1200$, $3.45\,\mu$m pitch) located one focal length away from the second lens. Linear polarizers placed before and after the beam splitter ensure the correct polarization state for phase modulation at the SLM. Phase ramps and global quadrature phase shifts displayed on the SLM produce calibrated spatially shifted interferometric measurements.

\begin{table}[t]
\centering
\caption{Optical hardware used in the real experiment setup.}
\label{supp:tab:hardware}
\resizebox{\linewidth}{!}{
\begin{tabular}{lccc}
\hline
\textbf{Component} & \textbf{Model} & \textbf{Quantity} & \textbf{Notes} \\
\hline
Laser diode (532 nm) & CPS532-C2 & 1 & Collimated illumination source \\
Spatial Light Modulator (SLM) & Holoeye LUNA VIS 111 & 1 & $1920\times1080$, $4.5\,\mu$m pitch \\
FLIR Camera & BFS-U3-23S3M-C & 1 & $1920\times1200$, $3.45\,\mu$m pitch \\
Non-polarizing Beam splitter & Thorlabs CCM1-BS013 & 1 & Routes reference and modulated beams \\
Lens & Thorlabs AC254-100-A & 2 & $f=100$ mm, $\varnothing 25.4$ mm \\
Diffuser & Thorlabs DG10-120-A & 1 & Glass diffuser \\
Aperture & ID12/M
 & 1 & Iris, 12 mm Max Aperture \\
Linear polarizers & - & 2 & Input/output polarization control \\
\hline
\end{tabular}
}
\end{table}

\begin{figure}[!ht]
    \centering
   \includegraphics[width=\linewidth]{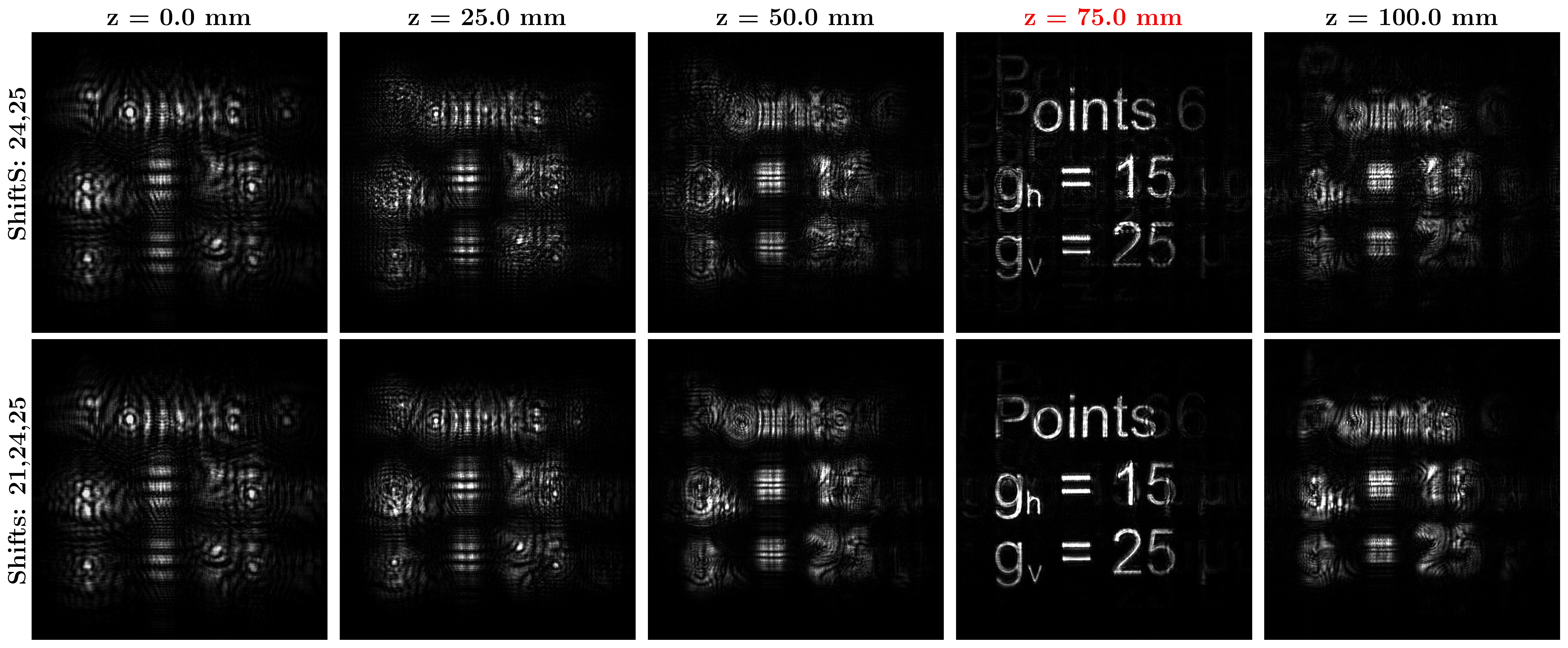}
    \rule{\linewidth}{0.4pt}
    \includegraphics[width=\textwidth]{figures/compressed/target_focusing_star_propagation_zoom_v2.jpg}
    \rule{\linewidth}{0.4pt}
    \includegraphics[width=\textwidth]{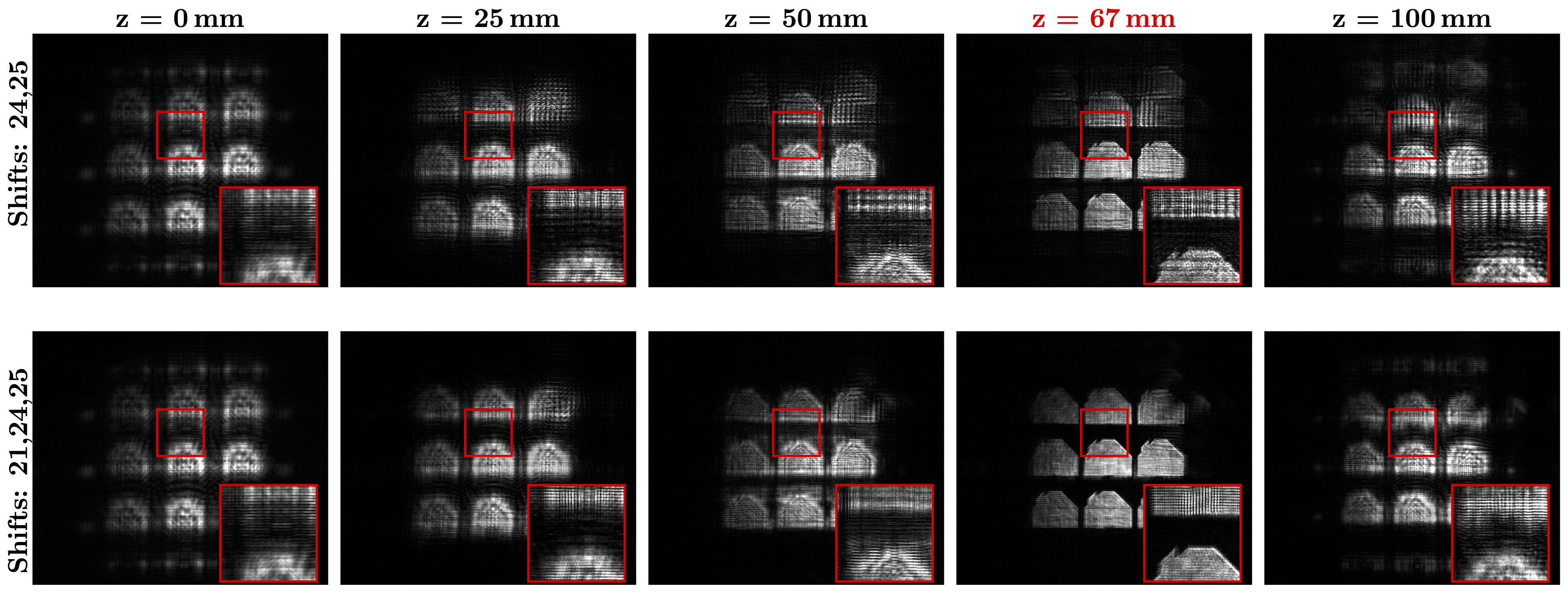}    
   \caption{Computational refocusing of the recovered wavefront for three different scenes using free-space propagation. The reconstructed complex field is propagated from $z=0$ mm to $z=100$ mm away from the object plane. A sharp focused image appears in the fourth column for each image. }
    \label{supp:fig:target_refocusing}
\end{figure}

\noindent
\textbf{Additional refocusing experiments.} In \cref{supp:fig:target_refocusing}, we show the computational refocusing results by propagating the  wavefront reconstructed from two sets of measurements: 16 measurements using shift pairs $(24,25)$, and 24 measurements using shifts $(21,24,25)$. 
The wavefront reconstructed at $z=0$mm is numerically propagated from $z=0$mm to $z=100$mm. We obtain a sharp focus around the correct target plane at $z = 75$mm or $z=67$mm for three different scenes. 
The zoomed-in views highlight that fine structures and edges in the scene are clearly resolved at the correct focus plane.

\noindent
\textbf{Comparison experiments.} We next present visual comparisons of refocusing results for different reconstruction methods in~\cref{supp:fig:focus_comp_digits}. Details about the comparison methods are presented in \cref{supp:sec:comp_details}. We consider the same two measurement settings: 16 measurements using shift pairs $(24,25)$, and 24 measurements using shifts $(21,24,25)$. In all cases, the targets are refocused using the phases recovered by each method. Qualitatively, our method yields the best reconstructions across all targets. In addition, applying least-squares refinement further improves the sharpness and quality of the recovered structures.

\section{Comparison Methods and Experimental Details}
\label{supp:sec:comp_details}
We conducted comparison experiments against the following methods: gradient descent–based complex field recovery with random initialization (\textbf{GD-Rand}) and spectral initialization~\cite{netrapalli2013phase} (\textbf{GD-Spec}), a \textbf{PnP-FISTA} plug-and-play recovery method using a pretrained \textbf{DRUNet} denoiser~\cite{zhang2021plug} as a prior, \textbf{WISH}~\cite{wu2019wish}, and a Deep Image Prior (\textbf{DIP})–based approach~\cite{ulyanov2018deep}. All experiments use multiple shifted measurements, where the field is interfered with spatially shifted copies along different directions and captured at multiple shift amounts. In all experiments, the measured amplitude that is captured with no shifts is used during initialization. 

\begin{figure}[!t]
    \centering
    \includegraphics[width=\linewidth]{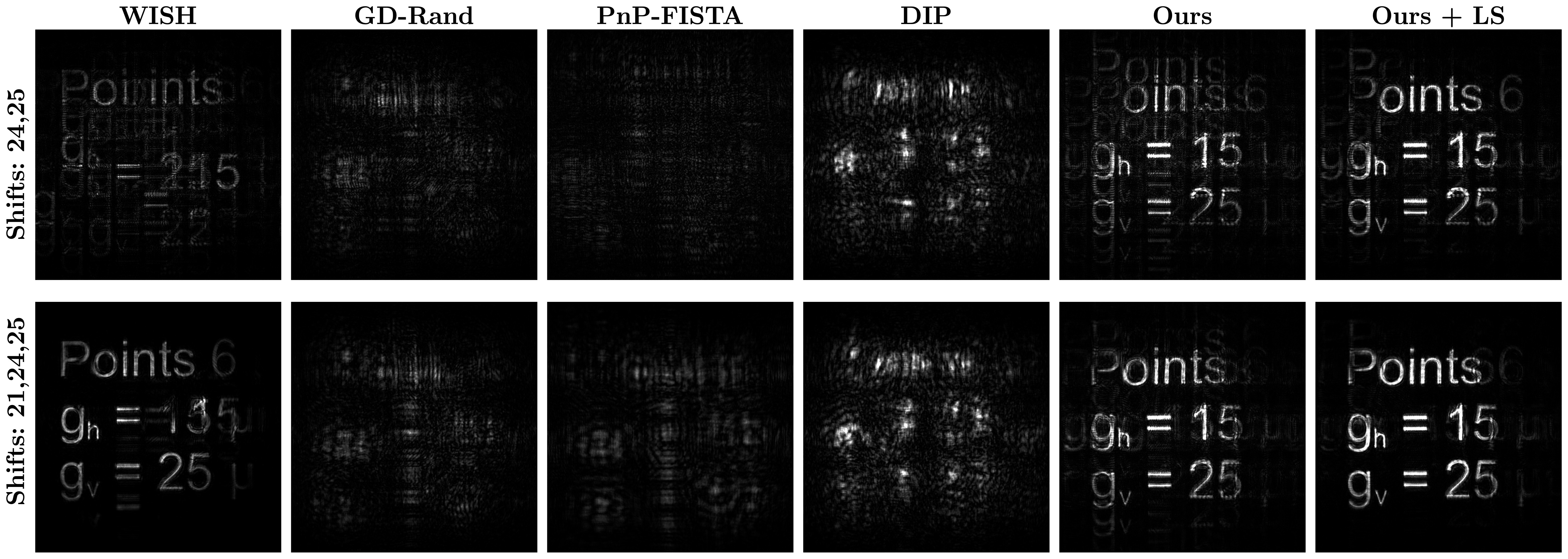}
    \rule{\linewidth}{0.4pt}
    \includegraphics[width=\linewidth]{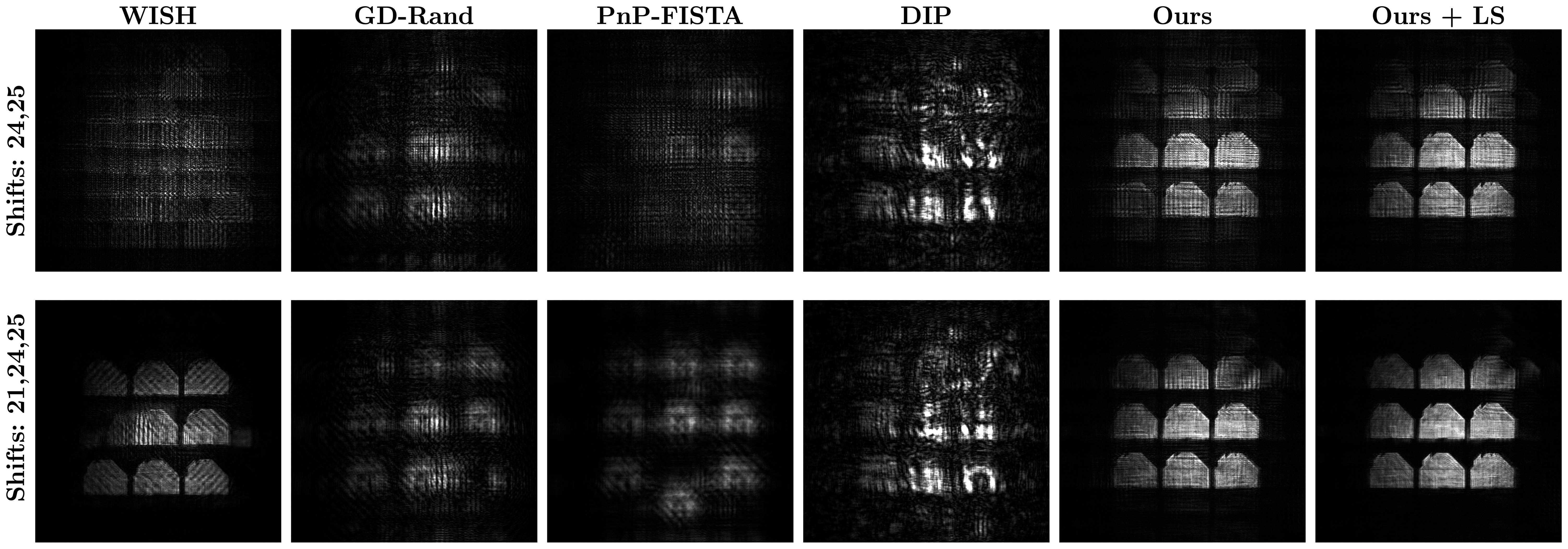}
   \rule{\linewidth}{0.4pt}
    \includegraphics[width=\linewidth]{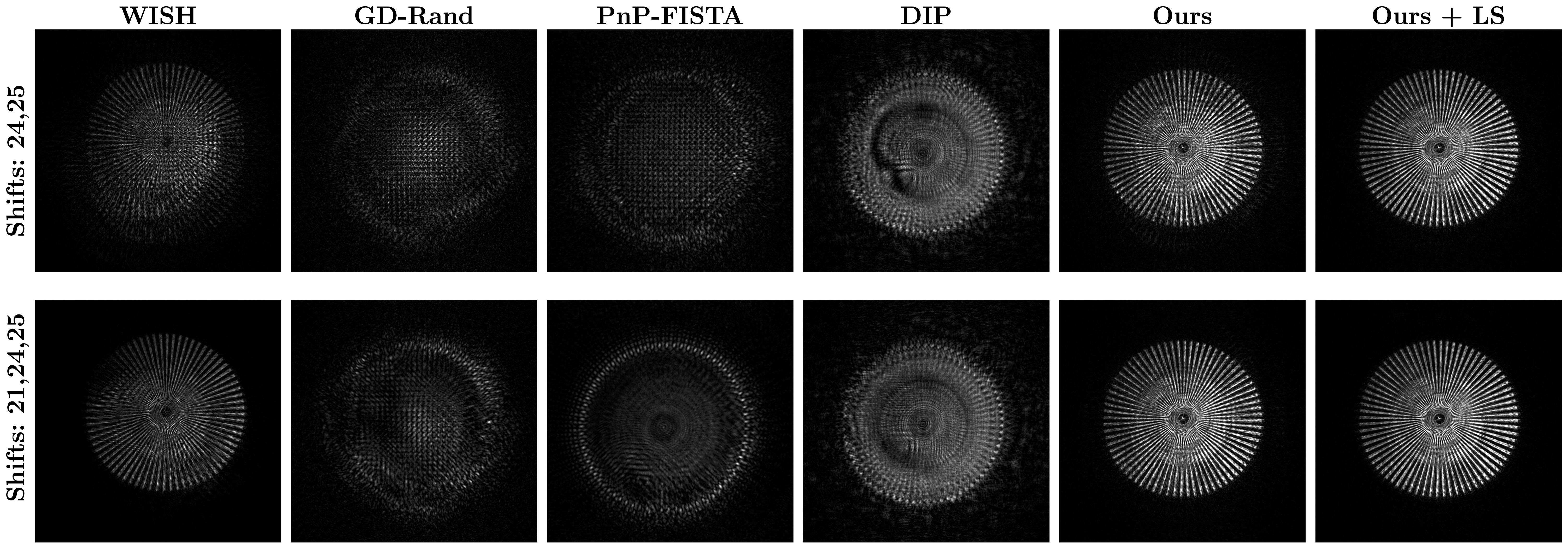}
    \caption{Comparison of refocused reconstructions for three  targets using different methods.}
    \label{supp:fig:focus_comp_digits}
\end{figure}

\begin{itemize}
      \item \textbf{Gradient-based methods.} We implement complex field recovery using gradient descent (GD) with stochastic gradient descent (SGD) in PyTorch, using a momentum of $0.9$ and running for 1000 iterations. The learning rate is set to $5\times 10^{-3}$ for all experiments. Two initialization strategies are used: random phase initialization (\textbf{GD-Rand}) and spectral initialization (\textbf{GD-Spec})~\cite{netrapalli2013phase}, where the initial field is computed using the adjoint operator from the DeepInverse library~\cite{tachella2025deepinverse}. The loss minimized is the mean-squared error between predicted and measured intensities.

       \item \textbf{PnP-FISTA.} We use the plug-and-play (PnP) framework for phase recovery using the DeepInverse library~\cite{tachella2025deepinverse}. The amplitude is initialized from the measured magnitude, and the phase is initialized randomly. A pretrained \textbf{DRUNet} denoiser~\cite{zhang2021plug} is used as the prior for the PnP framework. FISTA iterations run for 500 steps, with a stepsize of $8\times 10^{-2}$, and the denoiser strength $\sigma$ is set according to the measurement SNR.

     \item \textbf{WISH.} The complex field is recovered using the WISH algorithm~\cite{wu2019wish}. Each shifted interference measurement is modeled via a two-step propagation: a Fourier transform to the Fourier plane, phase modulation to simulate the shift, and an inverse Fourier transform, corresponding to a $4f$ optical system. The field at the Fourier plane is initialized by back-propagating the measured amplitudes with a random phase for each shift. The recovery algorithm follows a Gerchberg–Saxton-style alternating minimization: the algorithm iteratively enforces amplitude constraints at the measurement plane and updates the Fourier-plane estimate using a least-squares aggregation across all measurements. The process is repeated for 1000 iterations.
     
    \item \textbf{DIP.} We use the Deep Image Prior (DIP) framework~\cite{ulyanov2018deep} to recover the phase of the complex field. The phase is parameterized using a convolutional decoder network~\cite{darestani2021accelerated} that maps a fixed random latent input of size $16\times16$ to the output phase, while the amplitude is fixed to the measured amplitude. The network is optimized using Adam~\cite{kingma2014adam} with an initial learning rate of $10^{-3}$ for 1000 iterations, minimizing the mean-squared error between the predicted and measured intensities over all shifted measurements.

\end{itemize}

\section{Additional Simulation Experiments}
\label{supp:sec:sim_exp}

We report additional simulations to compare methods with different noise levels and numbers of measurements. To generate simulated data, we use 100 images from the DIV2K validation dataset~\cite{agustsson2017ntire} as ground truth intensity patterns. All images are center cropped to size $512 \times 512$. For each image, we synthesize complex wavefronts by combining the intensity with three types of phase profiles, including random phase, quadratic phase, and smooth phase profiles\footnote{The MATLAB \texttt{peaks} function is used to generate smooth synthetic profiles.}. We simulate shifted interferometric measurements with Poisson and Gaussian noise at multiple SNR levels and multiple shifts. We use the following shift configurations in the simulated experiments:

\begin{itemize}    \setlength\itemsep{0.4em} 

    \item 8 measurements: $\{\Delta_1\} = \{1\}$
    \item 16 measurements: $\{\Delta_1, \Delta_2\} = \{16,\,17\}$
    \item 32 measurements: $\{\Delta_i\}_{i=1}^{4} = \{16,\,17,\,23,\,31\}$
    \item 64 measurements: $\{\Delta_i\}_{i=1}^{8} = \{16,\,17,\,21,\,33,\,11,\,22,\,13,\,63\}$
\end{itemize}
Note that $(16,17)$ correspond to optimal shifts, whereas other shifts are selected at random.

Phase errors for quadratic phase profiles using $8,16,32,64$ measurements are reported in~\cref{supp:tab:phase_error_quad_snr_22,supp:tab:phase_error_quad_snr_13}, with visual reconstructions shown in \cref{supp:fig:quad_snr22,supp:fig:quad_snr13}. Across all measurement and noise settings our method performs the best. When using only $8$ measurements, our method significantly outperforms the competing approaches, especially at $22$ dB SNR. In addition, the least-squares refinement further improves the results in the $13$ dB SNR setting.

A similar pattern is observed for peak phase profiles (\cref{supp:tab:phase_error_peak_snr_22,supp:tab:phase_error_peak_snr_13}; \cref{supp:fig:peak_snr22,supp:fig:peak_snr13}). For random phase profiles (\cref{supp:tab:phase_error_rand_snr_22,supp:tab:phase_error_rand_snr_13}; \cref{supp:fig:random_snr22,supp:fig:random_snr13}), gradient descent achieves the best performance, while our method with least-squares refinement provides the second best results.

\begin{table}[htbp]
\centering
\small
\setlength{\tabcolsep}{3pt}
\caption{Phase errors for \textbf{Quadratic} phase recovery (mean $\pm$ std) at \textbf{22\,dB SNR}. Our method achieves the best performance. (Best bold, second-best underlined).}
\label{supp:tab:phase_error_quad_snr_22}
\resizebox{0.8\textwidth}{!}{
\begin{tabular}{lcccc}
\toprule

\diagbox{Method}{\#Meas.} & 8 & 16 & 32 & 64 \\
\midrule
GD-Rand & 1.536 $\pm$ 0.024 & 0.695 $\pm$ 0.365 & 0.470 $\pm$ 0.434 & 0.253 $\pm$ 0.385 \\
GD-Spec & 1.479 $\pm$ 0.063 & 0.737 $\pm$ 0.400 & 0.400 $\pm$ 0.432 & 0.358 $\pm$ 0.412 \\
PnP-FISTA & 1.505 $\pm$ 0.036 & 0.669 $\pm$ 0.341 & 0.456 $\pm$ 0.304 & 0.342 $\pm$ 0.225 \\
WISH & 1.524 $\pm$ 0.022 & 0.852 $\pm$ 0.299 & 0.423 $\pm$ 0.366 & 0.302 $\pm$ 0.401 \\
DIP & 1.541 $\pm$ 0.024 & 1.486 $\pm$ 0.031 & 1.491 $\pm$ 0.051 & 1.425 $\pm$ 0.071 \\

\hdashline \noalign{\vskip 0.5ex}
Ours & \underline{0.699 $\pm$ 0.247} & \textbf{0.158 $\pm$ 0.104} & \textbf{0.099 $\pm$ 0.083} & \textbf{0.076 $\pm$ 0.079} \\
Ours + LS & \textbf{0.623 $\pm$ 0.277} & \underline{0.183 $\pm$ 0.082} & \underline{0.126 $\pm$ 0.065} & \underline{0.112 $\pm$ 0.071} \\
\bottomrule
\end{tabular}
}
\end{table}

\begin{table}[htbp]
\centering
\small
\setlength{\tabcolsep}{3pt}
\caption{Phase errors for \textbf{Quadratic} phase recovery (mean $\pm$ std) at \textbf{13\,dB SNR}. Our method achieves the best performance and least-square refinements improve performance. (Best bold, second-best underlined).}
\label{supp:tab:phase_error_quad_snr_13}
\resizebox{0.8\textwidth}{!}{
\begin{tabular}{lllll}
\toprule

\diagbox{Method}{\#Meas.} & 8 & 16 & 32 & 64 \\
\midrule
GD-Rand & 1.539 $\pm$ 0.010 & 0.729 $\pm$ 0.361 & 0.471 $\pm$ 0.422 & 0.266 $\pm$ 0.369 \\
GD-Spec & 1.486 $\pm$ 0.051 & 0.859 $\pm$ 0.417 & 0.384 $\pm$ 0.285 & 0.370 $\pm$ 0.312 \\
PnP-FISTA & 1.522 $\pm$ 0.027 & 0.936 $\pm$ 0.304 & 0.518 $\pm$ 0.305 & 0.398 $\pm$ 0.267 \\
WISH & 1.536 $\pm$ 0.017 & 1.012 $\pm$ 0.276 & 0.500 $\pm$ 0.360 & 0.377 $\pm$ 0.405 \\
DIP & 1.541 $\pm$ 0.022 & 1.485 $\pm$ 0.030 & 1.485 $\pm$ 0.053 & 1.415 $\pm$ 0.082 \\

\hdashline \noalign{\vskip 0.5ex}
Ours & \underline{1.342 $\pm$ 0.161} & \underline{0.577 $\pm$ 0.311} & \underline{0.350 $\pm$ 0.240} & \underline{0.254 $\pm$ 0.184} \\
Ours + LS & \textbf{1.251 $\pm$ 0.225} & \textbf{0.445 $\pm$ 0.287} & \textbf{0.274 $\pm$ 0.194} & \textbf{0.200 $\pm$ 0.124} \\

\bottomrule
\end{tabular}
}
\end{table}

\begin{table}[htbp]
\centering
\small
\setlength{\tabcolsep}{3pt}
\caption{Phase errors for \textbf{Peak} phase recovery (mean $\pm$ std) at \textbf{22\,dB SNR}. Our method achieves the best results across all settings. (Best bold, second-best underlined).}
\label{supp:tab:phase_error_peak_snr_22}
\resizebox{0.8\textwidth}{!}{
\begin{tabular}{lcccc}
\toprule

\diagbox{Method}{\#Meas.} & 8 & 16 & 32 & 64 \\
\midrule
GD-Rand & 1.478 $\pm$ 0.021 & 0.791 $\pm$ 0.361 & 0.570 $\pm$ 0.404 & 0.286 $\pm$ 0.363 \\
GD-Spec & 1.468 $\pm$ 0.075 & 0.686 $\pm$ 0.397 & 0.371 $\pm$ 0.374 & 0.209 $\pm$ 0.367 \\
PnP-FISTA & 1.506 $\pm$ 0.035 & 0.663 $\pm$ 0.341 & 0.479 $\pm$ 0.338 & 0.373 $\pm$ 0.296 \\
WISH & 1.517 $\pm$ 0.024 & 0.898 $\pm$ 0.330 & 0.425 $\pm$ 0.367 & 0.309 $\pm$ 0.391 \\
DIP & 1.471 $\pm$ 0.022 & 1.442 $\pm$ 0.063 & 1.423 $\pm$ 0.081 & 1.352 $\pm$ 0.112 \\

\hdashline \noalign{\vskip 0.5ex}
Ours & \underline{0.665 $\pm$ 0.238} & \textbf{0.155 $\pm$ 0.096} & \textbf{0.097 $\pm$ 0.071} & \textbf{0.071 $\pm$ 0.063} \\
Ours + LS & \textbf{0.559 $\pm$ 0.247} & \underline{0.205 $\pm$ 0.063} & \underline{0.159 $\pm$ 0.050} & \underline{0.124 $\pm$ 0.050} \\

\bottomrule
\end{tabular}
}
\vspace{-1em}
\end{table}

\begin{table}[htbp]
\centering
\small
\setlength{\tabcolsep}{3pt}
\caption{Phase errors for \textbf{Peak} phase recovery (mean $\pm$ std) at \textbf{13\,dB SNR}. Our method achieves the best results across all settings. (Best bold, second-best underlined).}
\label{supp:tab:phase_error_peak_snr_13}
\resizebox{0.8\textwidth}{!}{
\begin{tabular}{lcccc}
\toprule

\diagbox{Method}{\#Meas.} & 8 & 16 & 32 & 64 \\
\midrule
GD-Rand & 1.485 $\pm$ 0.012 & 0.830 $\pm$ 0.338 & 0.593 $\pm$ 0.401 & 0.314 $\pm$ 0.358 \\
GD-Spec & 1.482 $\pm$ 0.054 & 0.800 $\pm$ 0.381 & 0.415 $\pm$ 0.361 & \underline{0.216 $\pm$ 0.195} \\
PnP-FISTA & 1.525 $\pm$ 0.022 & 0.948 $\pm$ 0.316 & 0.535 $\pm$ 0.327 & 0.377 $\pm$ 0.279 \\
WISH & 1.534 $\pm$ 0.018 & 1.019 $\pm$ 0.278 & 0.552 $\pm$ 0.375 & 0.408 $\pm$ 0.404 \\
DIP & 1.470 $\pm$ 0.022 & 1.452 $\pm$ 0.054 & 1.431 $\pm$ 0.074 & 1.360 $\pm$ 0.103 \\

\hdashline \noalign{\vskip 0.5ex}
Ours & \underline{1.307 $\pm$ 0.194} & \underline{0.532 $\pm$ 0.277} & \underline{0.342 $\pm$ 0.227} & 0.255 $\pm$ 0.200 \\
Ours + LS & \textbf{1.181 $\pm$ 0.276} & \textbf{0.394 $\pm$ 0.234} & \textbf{0.273 $\pm$ 0.175} & \textbf{0.206 $\pm$ 0.152} \\
\bottomrule
\end{tabular}
}
\vspace{-1em}
\end{table}

\begin{table}[htbp]
\centering
\small
\setlength{\tabcolsep}{3pt}
\caption{Phase errors for \textbf{Random} phase recovery (mean $\pm$ std) at \textbf{22\,dB SNR}. Our method consistently achieves the best or second-best results. (Best bold, second-best underlined).}
\label{supp:tab:phase_error_rand_snr_22}
\resizebox{0.8\textwidth}{!}{
\begin{tabular}{lllll}
\toprule

\diagbox{Method}{\#Meas.} & 8 & 16 & 32 & 64 \\
\midrule
GD-Rand & \underline{0.505 $\pm$ 0.142} & \textbf{0.138 $\pm$ 0.148} & \textbf{0.079 $\pm$ 0.152} & \textbf{0.043 $\pm$ 0.126} \\
GD-Spec & 1.481 $\pm$ 0.064 & 0.691 $\pm$ 0.367 & 0.449 $\pm$ 0.440 & 0.148 $\pm$ 0.203 \\
PnP-FISTA & 1.511 $\pm$ 0.030 & 0.735 $\pm$ 0.329 & 0.425 $\pm$ 0.264 & 0.364 $\pm$ 0.233 \\
WISH & 1.519 $\pm$ 0.027 & 0.854 $\pm$ 0.320 & 0.453 $\pm$ 0.377 & 0.285 $\pm$ 0.377 \\
DIP & 1.423 $\pm$ 0.063 & 1.495 $\pm$ 0.037 & 1.447 $\pm$ 0.065 & 1.384 $\pm$ 0.104 \\

\hdashline \noalign{\vskip 0.5ex}
Ours & 0.659 $\pm$ 0.236 & \underline{0.156 $\pm$ 0.099} & \underline{0.094 $\pm$ 0.075} & \underline{0.073 $\pm$ 0.066} \\
Ours + LS & \textbf{0.492 $\pm$ 0.273} & 0.158 $\pm$ 0.075 & 0.121 $\pm$ 0.055 & 0.107 $\pm$ 0.054 \\
\bottomrule
\end{tabular}
}
\end{table}

\begin{table}[htbp]
\centering
\small
\setlength{\tabcolsep}{3pt}
\caption{Phase errors for \textbf{Random} phase recovery (mean $\pm$ std) at \textbf{13\,dB SNR}. Our method consistently achieves the best or second-best results. (Best bold, second-best underlined).}
\label{supp:tab:phase_error_rand_snr_13}
\resizebox{0.8\textwidth}{!}{
\begin{tabular}{lcccc}
\toprule

\diagbox{Method}{\#Meas.} & 8 & 16 & 32 & 64 \\
\midrule
GD-Rand & \textbf{0.642 $\pm$ 0.210} & \textbf{0.215 $\pm$ 0.205} & \textbf{0.126 $\pm$ 0.159} & \textbf{0.092 $\pm$ 0.173} \\
GD-Spec & 1.485 $\pm$ 0.055 & 0.842 $\pm$ 0.383 & 0.495 $\pm$ 0.377 & 0.360 $\pm$ 0.331 \\
PnP-FISTA & 1.522 $\pm$ 0.027 & 0.951 $\pm$ 0.340 & 0.595 $\pm$ 0.339 & 0.390 $\pm$ 0.275 \\
WISH & 1.533 $\pm$ 0.021 & 0.995 $\pm$ 0.277 & 0.534 $\pm$ 0.362 & 0.370 $\pm$ 0.385 \\
DIP & 1.426 $\pm$ 0.061 & 1.499 $\pm$ 0.036 & 1.446 $\pm$ 0.068 & 1.368 $\pm$ 0.086 \\

\hdashline \noalign{\vskip 0.5ex}
Ours & 1.283 $\pm$ 0.173 & 0.555 $\pm$ 0.299 & 0.360 $\pm$ 0.246 & 0.254 $\pm$ 0.176 \\
Ours + LS & \underline{1.173 $\pm$ 0.270} & \underline{0.407 $\pm$ 0.266} & \underline{0.283 $\pm$ 0.208} & \underline{0.195 $\pm$ 0.125} \\

\bottomrule
\end{tabular}
}
\end{table}

\begin{figure}[ht]
    \centering
    \includegraphics[width=0.95\linewidth]{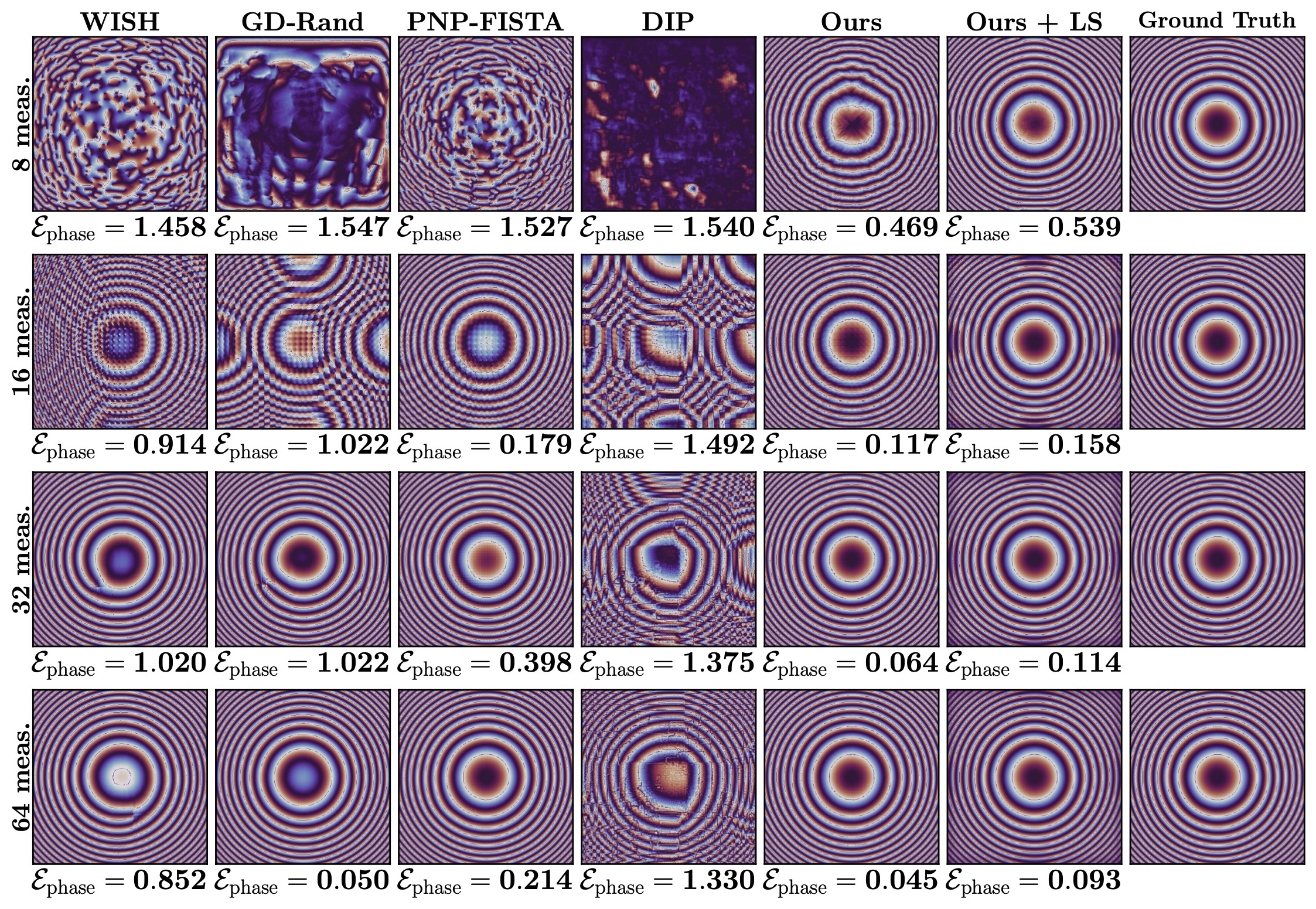}
    \caption{Quadratic phase recovery at different number of measurements at $22$dB SNR. Our method consistently yields the best performance both quantitatively and qualitatively.}
    \label{supp:fig:quad_snr22}
\end{figure}

\begin{figure}[t]
    \centering
    \includegraphics[width=0.95\linewidth]{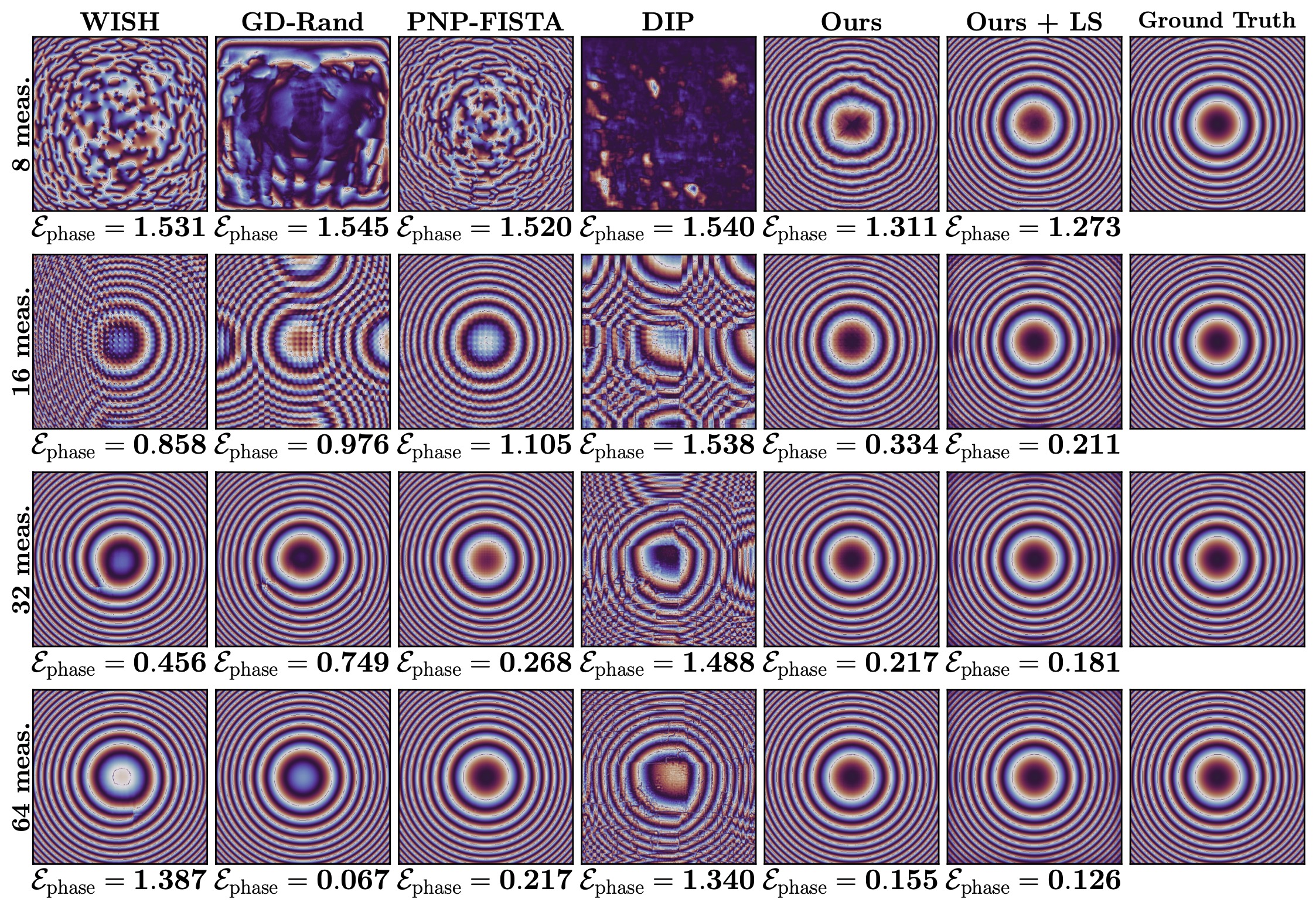}
    \caption{Quadratic phase recovery at $13$ dB SNR. Our method is the only approach that reliably reconstructs the phase profile using only $8$ measurements.}
    \label{supp:fig:quad_snr13}
    \vspace{-1em}
\end{figure}

\begin{figure}[ht]
    \centering
    \includegraphics[width=0.95\linewidth]{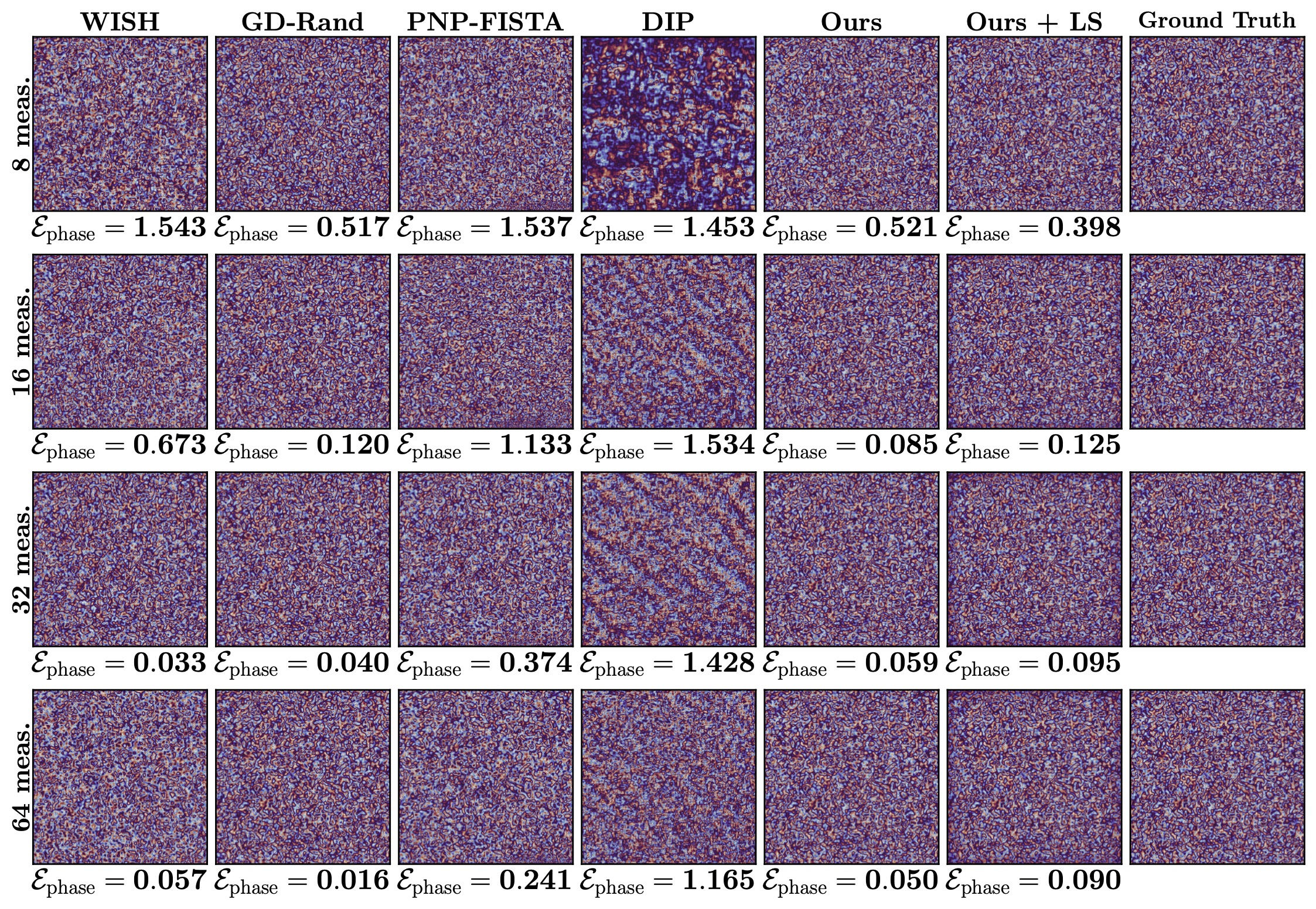}
    \caption{Random phase recovery at $22$ dB SNR with different numbers of measurements. Our method with least-squares refinement achieves the second-best performance across most measurement settings.}
    \label{supp:fig:random_snr22}
\end{figure}

\begin{figure}[ht]
    \centering
    \includegraphics[width=0.95\linewidth]{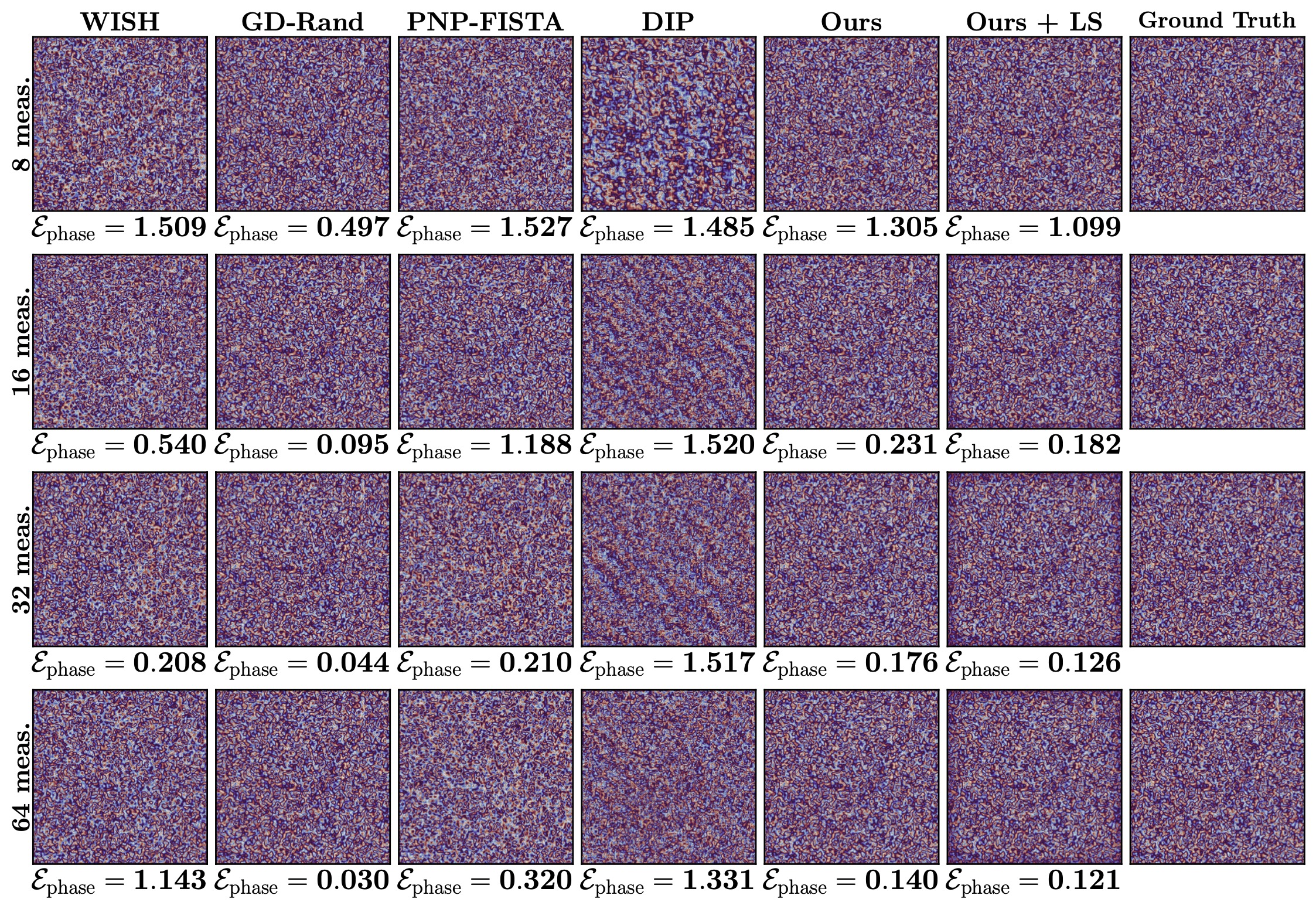}
    \caption{Random phase recovery at $13$ dB SNR with different numbers of measurements. Our method with least-squares refinement achieves the second-best performance across most measurement settings.}
    \label{supp:fig:random_snr13}
    \vspace{-1em}
\end{figure}

\begin{figure}[ht]
    \centering
    \includegraphics[width=0.95\linewidth]{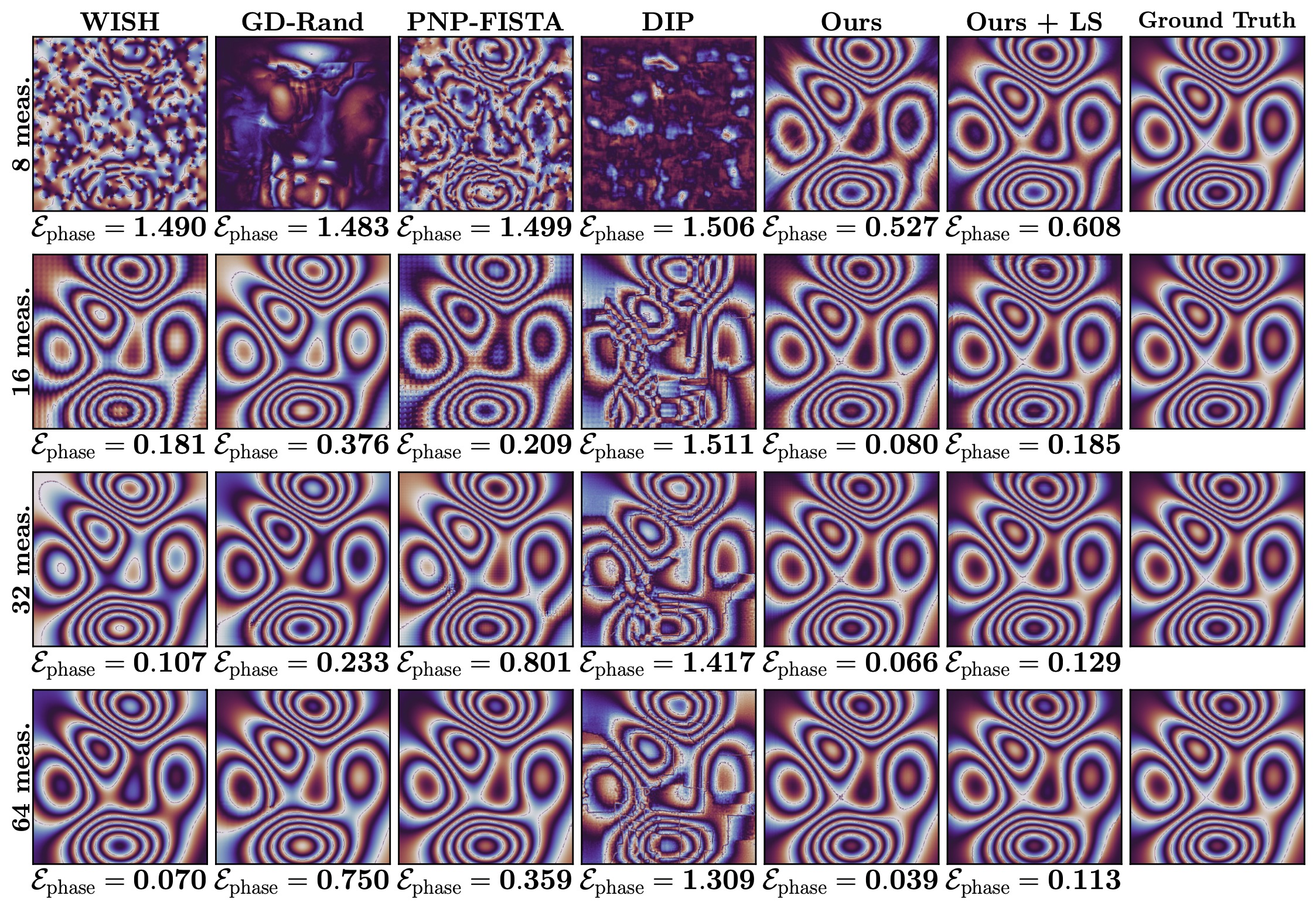}
    \caption{Peak phase recovery at $22$ dB SNR with different numbers of measurements. Our method achieved the best performance across all settings.}
    \label{supp:fig:peak_snr22}
     \vspace{-1em}
\end{figure}

\begin{figure}[ht]
    \centering
    \includegraphics[width=0.95\linewidth]{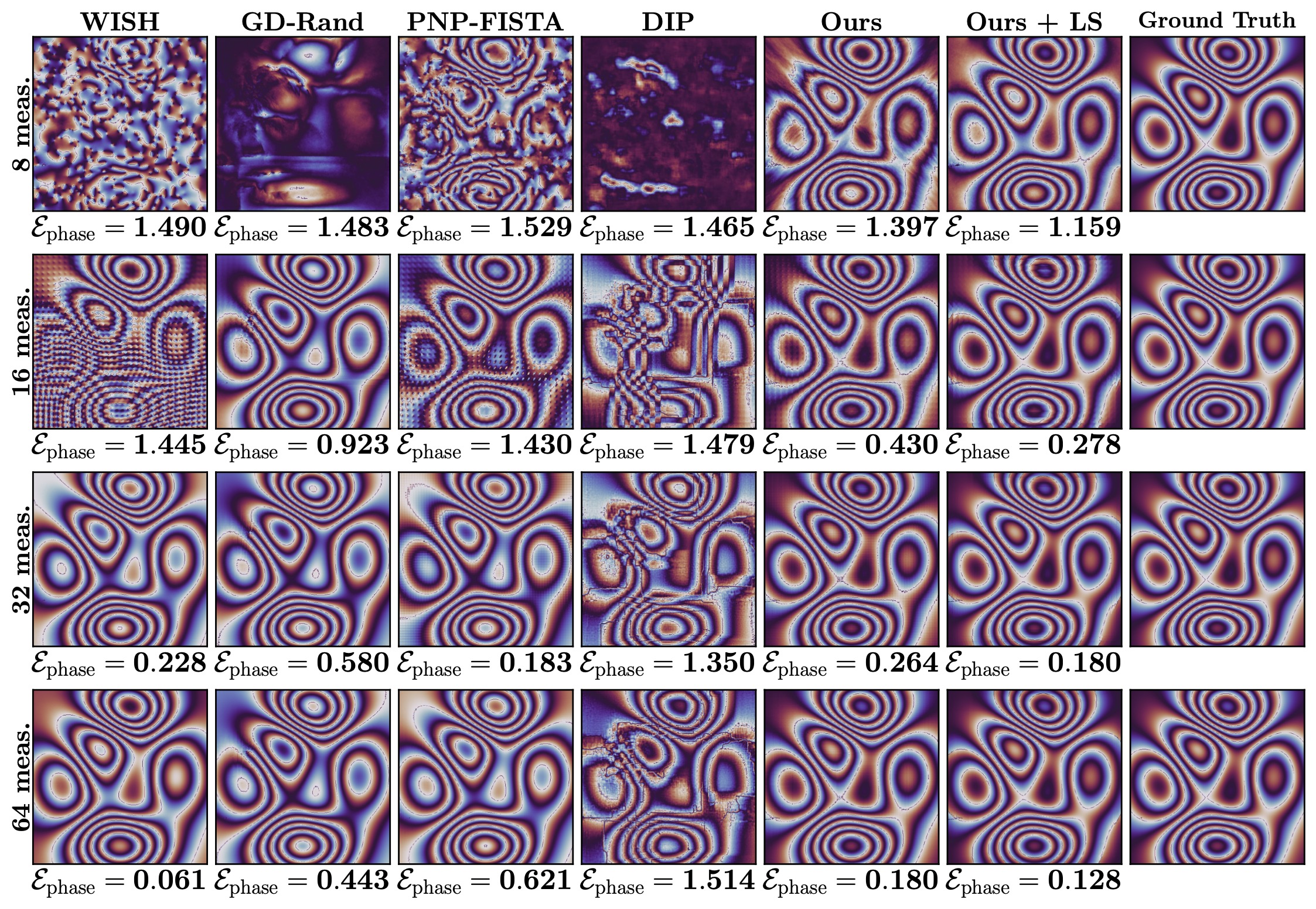}
    \caption{Peak phase recovery at $13$ dB SNR with different numbers of measurements. Our method achieves the best performance in most settings and is the only method that produces a consistent phase profile with very few measurements.}
    \label{supp:fig:peak_snr13}
\end{figure}

\end{document}